\documentclass[lettersize,journal]{IEEEtran}
\usepackage{amsmath,amsfonts}
\usepackage{mathrsfs}
\usepackage{amssymb}
\usepackage{algorithmic}
\usepackage{array}
\usepackage{subfig}

\usepackage{textcomp}
\usepackage{stfloats}
\usepackage{url}
\usepackage{verbatim}
\usepackage{graphicx}
\usepackage{algorithm}
\usepackage{dsfont}
\usepackage{amsbsy}
\usepackage{caption}
\usepackage{color}
\usepackage[normalem]{ulem}

\def\BibTeX{{\rm B\kern-.05em{\sc i\kern-.025em b}\kern-.08em
T\kern-.1667em\lower.7ex\hbox{E}\kern-.125emX}}
\usepackage{balance}

\newtheorem{lemma}{Lemma}

\newtheorem{proposition}{Proposition}

\newtheorem{remarks}{Remarks}
\newtheorem{discussion}{Discussion}

\newcommand{\bE}{\mathds{E}}
\newcommand{\mc}{\mathcal}
\newcommand{\mb}{\mathbf}
\newcommand{\bs}{\boldsymbol}
\newcommand{\ol}{\overline}
\newcommand{\wt}{\widetilde}
\newcommand{\wh}{\widehat}
\newcommand{\id}{\mathds{1}}
\newcommand\hc{OUTTA}
\newcommand\lc{PC-OUTTA}

\def\argmax{\operatornamewithlimits{arg\,max}}

\allowdisplaybreaks
\title{Efficient User Association and Wireless Scheduling with Shorter Time-Scale Rate Adaptation}
\author{Xiaoyi Wu,~\IEEEmembership{Graduate Student Member,~IEEE},
Huacheng Zeng,~\IEEEmembership{Senior Member,~IEEE},
and Bin Li,~\IEEEmembership{Senior Member,~IEEE}
\thanks{\copyright~2026 IEEE. Personal use of this material is permitted. Permission from IEEE must be obtained for all other uses, in any current or future media, including reprinting/republishing this material for advertising or promotional purposes, creating new collective works, for resale or redistribution to servers or lists, or reuse of any copyrighted component of this work in other works.}}

\begin{document}
\maketitle

\begin{abstract}
Rate adaptation is a crucial mechanism in IEEE 802.11 networks and next-generation cellular systems. Since the time scale for rate adaptation is typically much shorter than that for user association and scheduling, we investigate a joint design of wireless user association and scheduling and rate adaptation across different time scales to maximize cumulative network throughput while ensuring desired fairness among users. We develop a MaxWeight-type user association and scheduling algorithm that integrates virtual queues—tracking each user's scheduling debt to maintain fairness—and Upper Confidence Bound (UCB) estimates in its weight measure. Each selected user then employs the UCB algorithm for rate adaptation on a short time scale. Our theoretical findings reveal that the proposed algorithm achieves cumulative regret that grows with the square root of the time horizon up to a logarithmic factor and results in zero cumulative fairness violation after a certain number of time frames. Furthermore, since the MaxWeight-type algorithm involves evaluating all the feasible schedules that can be exponential to the number of users due to the interference constraints, leading to high computational complexity, we introduce a low-complexity alternative utilizing the so-called pick-and-compare (PC) approach. We demonstrate the effectiveness of both algorithms through simulations based on real-world data traces.
\end{abstract}

\begin{IEEEkeywords}
User Scheduling, Rate Adaptation, UCB, Fairness, Low-Complexity.
\end{IEEEkeywords}

\section{Introduction}
Multiple access points (APs) are commonly deployed in high-density areas such as campuses, stadiums, airports, subways, and shopping centers to ensure adequate communication capacity for reliable and fast wireless transmissions. Current APs are equipped with rate adaptation capability, enabling the transmitter to adjust the transmission rate using various channel coding and modulation schemes to accommodate the time-varying wireless channel, significantly enhancing system throughput. Moreover, the rate adaptation will be a key physical-layer mechanism for next-generation millimeter-wave (mmWave) communication systems that typically have a large and unpredictable throughput fluctuation  (see, e.g., \cite{narayanan2020lumos5g}). Rate adaptation typically occurs every 100~ms in IEEE 802.11 systems \cite{queiros2022wi}, while in mmWave-based systems, it operates on a much smaller time scale (e.g., less than 10~ms or even 1~ms). Such an operation time scale is typically smaller than each user's transmission session, necessitating a joint user association and scheduling design and rate adaptation. This design must determine when and which AP each user can associate with and is scheduled for wireless transmissions, as well as which transmission rate each selected user should choose on a small time scale during its scheduling period. The goal is to maximize network throughput while guaranteeing desired fairness among users (i.e., each user should be at least scheduled for a certain fraction of time on average).

User association and scheduling design is important for efficiently managing interference in wireless networks, which has been extensively studied (e.g., \cite{bejerano2004fairness,gong2008dynamic,li2013ap,athanasiou2014optimizing,dwijaksara2016joint,sun2017novel}). However, the integration of joint user association and scheduling design and rate adaptation with different time scales remains underexplored. In this paper, we consider user association and scheduling decisions are made every $T$ time slots. The user dissociation and association involve disconnection and reconnection, typically incurring non-trivial costs such as communication interruptions and increased network delays. While a small value of $T$ gives APs more flexibility to make association and scheduling decisions, potentially improving network performance, it also causes more frequent dissociation and association, resulting in high handover costs. Although frame-based scheduling designs (e.g., \cite{hou2013scheduling,jaramillo2011scheduling}) for deadline-constrained traffic share some similarities with our problem context, the user scheduling is not fixed within the entire frame and thus is fundamentally different from our problem.

On the other hand, rate adaptation is a key mechanism for wireless communication systems to approach wireless channel capacity and has been widely studied in the literature. Earlier works (e.g., \cite{bicket2005bit,kamerman1997wavelan,lacage2004ieee}) developed heuristic rate selection algorithms to balance exploration and exploitation. In \cite{combes2018optimal}, the authors formulated rate adaptation as a Multi-Armed Bandit (MAB) problem, where each arm corresponds to a rate associated with an unknown link successful transmission probability. As such, all classical MAB algorithms, such as Upper Confidence Bound (UCB \cite{auer2002finite}), Kullback-Leibler UCB (KL-UCB \cite{garivier2011kl}), and Thompson Sampling \cite{agrawal2012analysis}, can be directly applied to develop rate adaptation algorithms with provable performance guarantees. The authors in \cite{combes2018optimal} further developed a KL-UCB-based rate adaptation algorithm by exploiting the unimodal structure of the system throughput with respect to the transmission rate. Subsequent works \cite{gupta2018low,gupta2019link} further developed more efficient rate adaptation algorithms based on Thompson Sampling. However, all these rate adaptation algorithms focused on a single wireless link and have not yet been integrated into the wireless user scheduling design that operates on a much larger time scale.

In this paper, we develop a joint wireless user association and scheduling and rate adaptation algorithm with different operating time scales. The goal is to maximize the cumulative throughput over a finite time horizon while guaranteeing the desired fairness among users, i.e., each user is scheduled at least a certain fraction of time on average. This joint algorithm design is similar to the combinatorial bandits with fairness constraints (e.g., \cite{li2019combinatorial,liu2021efficient}), where each arm corresponds to a user and fairness is ensured among users. In particular, \cite{li2019combinatorial} introduced the virtual queues to address fairness constraints and incorporated it into the algorithm design that yields a cumulative regret growing with the square root of the time horizon up to a logarithmic factor while guaranteeing long-term fairness among users. \cite{liu2021efficient} developed a pessimistic-optimistic algorithm with regret that grows with the square root of the time horizon and agrees with the state-of-the-art instance-independent bound, and zero cumulative fairness violation after a certain time. 

Our problem setup differs from those works in the following aspects. First, we consider joint user association, scheduling, and rate adaptation with fairness constraints, as opposed to user scheduling in \cite{li2019combinatorial,liu2021efficient}. Second, our user association and scheduling do not rely solely on MAB-based online learning but instead utilize the MaxWeight-type algorithm, embedding an appropriate MAB component for each user in the joint user association and scheduling and rate adaptation framework. Third, user association and scheduling and rate adaptation operate at different time scales, posing unique challenges for algorithm design and theoretical analysis. To the best of our knowledge, this is the first work addressing wireless user association and scheduling and rate adaptation in different time scales.

The main contributions of this work are summarized as follows: 

\begin{itemize}
    \item 
We develop an online-learning-based joint user association and scheduling and rate adaptation that operate on different time scales (cf. Section \ref{sec:alg}). In particular, a MaxWeight-type algorithm with the weight combining the virtual queues and UCB estimates is utilized to determine the user association and scheduling, while the UCB algorithm is employed to determine the transmission rate of each selected user. 

\item
We show that our proposed algorithm achieves $O(\sqrt{K\log K})$ regret over $K$ time frames while zero cumulative fairness violation can be achieved after a certain number of time frames independent of $K$.

\item 
We further propose a low-complexity algorithm using the pick-and-compare approach, reducing the computational complexity of the MaxWeight-type joint user association and scheduling and rate adaptation algorithm.

\item 
We demonstrate the superior performance of our proposed algorithms via simulations using the experimental data collected in realistic wireless networks.
\end{itemize}

While this paper is built upon our prior conference version \cite{wu2023joint}, which was published in Wiopt 2023, we have the following
new contributions: (1) we conduct an extensive literature survey related to our research; (2) we adopt the pick-and-compare method to significantly reduce the computational complexity of the MaxWeight-type joint user association and scheduling and rate adaptation algorithm; (3) we add simulations to verify the effectiveness of this low-complexity algorithm using real-world data; (4) more detailed proofs for Proposition \ref{prop:violation} and Proposition \ref{prop:regret} are included.

\emph{Note on Notation}: We use bold and script font of a variable to denote a vector and a set, respectively. 
Let $\|\mb{x}\|_1$ and $\|\mb{x}\|$ denote the $l_1$ and $l_2$ norm of the vector $\mb{x}$, respectively. Let $f(x)=O(x)$ if $f(x)\leq Cx, \forall x\geq0$ for some positive real number $C$.

\section{related work}
In this section, we survey two main areas that are closely related to our work: rate adaptation and MAB, and wireless scheduling with fairness constraints.
\subsection{Rate Adaptation and MAB}
Rate adaptation is a crucial strategy for wireless communication systems, which ensures optimal performance by adapting to time-varying channel conditions, interferences, and other dynamics. This technique has gained significant attention in research. 
The authors in \cite{combes2018optimal} redefined rate adaptation using the framework of MAB. In this context, each ``arm'' represents a specific transmission rate, each having an uncertain probability of successful link transmission. Different from traditional unstructured MAB problems, the authors in \cite{combes2018optimal} also explored the unimodal structure of the rate adaptation problem, which can be exploited to identify the optimal rate faster. They further applied a modified version of KL-UCB that takes advantage of this unimodal structure. The authors in 
\cite{gupta2018low,gupta2019link}
developed more efficient rate adaptation algorithms based on Thompson Sampling.
Subsequent works (e.g., \cite{qureshi2019fast, qureshi2020online}) introduced contextual learning to the rate adaptation problem under the MAB framework. Authors in \cite{qureshi2019fast} proposed a contextual learning algorithm based on KL-UCB, which exploits the unimodality of the expected reward both in the arms and the contexts. Additionally, recent studies (e.g., \cite{bharatula2023adapting, tong2023rate}) have capitalized on the correlations in the rewards associated with different transmission rates to estimate one rate by suitably combining the observed rewards of other rates.
Other works, such as \cite{zhu2019joint,tong2020optimal}, considered the rate adaptation jointly with other decision problems using the framework of MAB.
However, all these rate adaptation algorithms have not yet been integrated into the wireless user scheduling design that operates on a much larger time scale.     

\subsection{Wireless Scheduling with Constraints}
Wireless scheduling involves determining which devices or users get access to the network resources. The design of multi-user wireless schedulers has received substantial attention. For infinitely backlogged user queues, researchers have devised various schedulers that can achieve throughput optimality, such as MaxWeight rule \cite{andrews2004scheduling,tassiulas1993dynamic}, Exp
rule \cite{shakkottai2002scheduling} and Log rule \cite{sadiq2010delay}. Subsequently, the authors in \cite{liu2001opportunistic} introduced fairness constraints to wireless scheduling, which ensures each user should be served with a fraction of time on average. However, they did not address the algorithm design in an unknown wireless network environment. Max-min fairness has also been studied from a joint rate-control and power-allocation perspective. For example, the authors in~\cite{zheng2017max} characterize the optimal weighted max-min rate fairness under interference coupling via nonlinear Perron-Frobenius theory. In contrast, our work enforces fairness on the scheduling fractions through virtual queues and learns the unknown transmission-success statistics online. The notion of an approximation ratio has a long history in wireless resource allocation, where it quantifies the loss incurred by replacing an exact but computationally expensive allocation rule with a cheaper one. At the physical layer, the authors in~\cite{tan2013fast} develop fast algorithms for sum rate maximization together with performance bounds that certify how close the returned solution is to the optimum, which turns an intractable problem into one with a quantified optimality gap. At the scheduling layer, the authors in~\cite{lin2005impact} show that an imperfect scheduler attaining a constant fraction of the maximum weight stabilizes a correspondingly scaled version of the capacity region, and the pick-and-compare rule of~\cite{tassiulas1997scheduling} recovers the full region while comparing a single schedule per slot. In all of these results the channel statistics are known in advance, so the approximation ratio is a fixed constant determined by the accuracy of the allocation rule. A related line reaches a finite-horizon guarantee by a different route, either by competitive analysis against the offline optimum~\cite{lin2013dynamic} or by a controlled approximation of the underlying combinatorial problem~\cite{chen2013markov}, where the uncertainty concerns the input rather than the system statistics. What all of these share is that nothing has to be learned, so no exploration cost enters the analysis and the ratio is a horizon-free constant. In our setting the per-link transmission-success statistics are unknown and must be learned online, so the algorithm resolves an exploration-exploitation trade-off whose cost is the cumulative regret, and the approximation ratio is horizon-dependent rather than a fixed constant. We make this connection precise in Discussion~\ref{disc:approx} and evaluate it on real-world traces in Section~\ref{sec:sim}.
Our problem of wireless user scheduling with fairness constraints under unknown system statistics can be analogous to the combinatorial bandits with fairness constraints (e.g., \cite{li2019combinatorial,liu2021efficient,patil2021achieving,liucombinatorial,bernasconi2022safe}), where each arm should be pulled for at least a certain minimum fraction of times in addition to the objective of maximizing the sum of expected rewards without the knowledge of reward distributions in advance. The authors in \cite{li2019combinatorial} introduced virtual queues to track the fairness violation and incorporated them into the algorithm design together with the UCB weight \cite{auer2002finite} for estimating the rewards. They characterized cumulative regret and long-term fairness performance. Recently, the authors in \cite{liu2021efficient} proposed a pessimistic-optimistic algorithm that achieves state-of-the-art regret performance and a zero fairness constraint violation by properly selecting algorithmic parameters. However, none of these works can be directly applied to our context, as our wireless user scheduling should be operated on different time scales, resulting in distinct dynamics for the virtual queue-length and UCB weight.
Although there are similarities with frame-based scheduling methods for deadline-constrained traffic (e.g., \cite{hou2013scheduling,jaramillo2011scheduling}), the key difference is that user scheduling is not fixed across the entire frame.

\section{system model}
\begin{figure}[t]
    \centering
    \includegraphics[width=\linewidth]{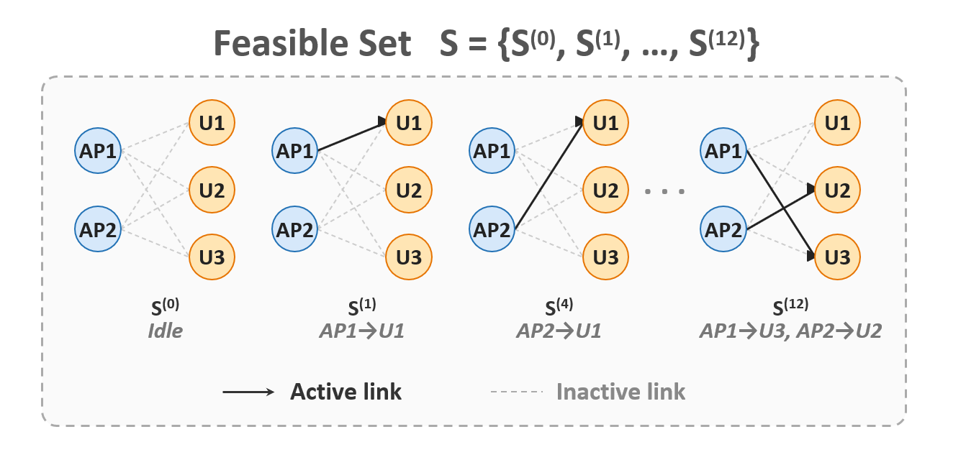}
    \caption{ Illustration of the feasible set $\mathcal{S} = \{\mathbf{S}^{(0)}, \mathbf{S}^{(1)}$, $\ldots, \mathbf{S}^{(12)}\}$ for a network with $L=2$ APs and $N=3$ users, where each user is served by at most one AP per frame. Each element $\mathbf{S}^{(i)} \in \mathcal{S}$ represents a feasible schedule.}
    \label{fig:feasible-set}
\end{figure}

We consider a wireless system with $N$ users and $L$ APs. As shown in Fig.~\ref{fig:system-architecture}, time is divided into frames, each consisting of $T$ time slots, and the total time horizon spans $K$ such frames. At the beginning of each frame (i.e., at time indices $\{0, T, 2T, \ldots, (K-1)T\}$), the system determines which user--AP pairs to activate. Each selected user is both associated with and scheduled by the assigned AP for the entire duration of the frame, i.e., for $T$ time slots. These assignments remain fixed within the frame and are updated only at frame boundaries.
A smaller frame size $T$ enables the system to reassess user--AP pairings more frequently in response to time-varying channel conditions, thereby allowing users to connect to better APs more often. However, this increased flexibility comes at the cost of more frequent handovers and their associated overhead. Thus, the choice of frame size $T$ reflects a trade-off between network adaptability and handover cost. Due to the wireless interference constraints, only a subset of users can transmit simultaneously in each time frame. 
We define $S_{l,n}(kT)=1$ if user $n$ is associated with AP $l$ for wireless transmission in frame $k$, and $S_{l,n}(kT)=0$ otherwise. The red dots represent the \emph{decision variable} $S_{l,n}(kT)$, as illustrated in Fig.~\ref{fig:system-architecture}.
We refer to $\mb{S}(kT)\triangleq(S_{l,n}(kT),\forall l,\forall n)$ as the \emph{feasible schedule}, indicating the set of users that can be served by each AP simultaneously in frame $k$. Let $\mc{S}$
denote the collection of all feasible schedules that satisfy
the access constraints (e.g., each user associates with at
most one AP and each AP serves at most a certain number of
users per frame), as illustrated in
Fig.~\ref{fig:feasible-set} for a simple example with $L=2$
APs and $N=3$ users.

\begin{figure}[t]
  \centering
\includegraphics[width=0.5\textwidth,height=0.2\textwidth,keepaspectratio]{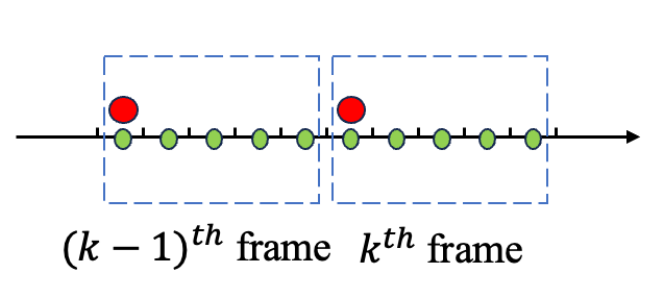}
  \caption{ An illustrative example of joint user association and scheduling with rate adaptation operating on different time scales. The frame is represented by the blue dash box with length pre-set to $T = 5$. For each fixed $(l,n,m)$, the red dots denote the \emph{decision varible} $S_{l,n}(kT)$, indicating whether user $n$ is associated and scheduled with AP $l$ at the beginning of the frame $k$, while the green dots denote the \emph{decision varible} $I_{l,n,m}(t)$, indicating whether user $n$ transmits to AP $l$ at rate $r_m$ in time slot $t$.}
  \label{fig:system-architecture}
  \vspace{-0.1in}
\end{figure}

Within each time slot of a frame, each selected user transmits to the AP at a rate chosen from the set $\{r_1,r_2,\dots,r_M\}$, where $0<r_1<r_2<\cdots<r_M$ and $M$ is the number of available transmission rates.
Let $X_{l,n,m}(t)=1$ indicate that the wireless transmission of user $n$ associated with AP $l$ at rate $r_m$ is successful in time slot $t$, and $X_{l,n,m}(t)=0$ otherwise. Here, $r_mX_{l,n,m}(t)$ represents the throughput when user $n$ associated with AP $l$ transmits at rate $r_m$ in time slot $t$. We assume that $X_{l,n,m}(t)$ is independently and identically distributed (i.i.d.) with an \emph{unknown} mean $\mu_{l,n,m}\in[0,1]$. We assume that for each fixed $(l, n, m)$, the sequence $\{X_{l,n,m}(t)\}_{t=0}^{T-1}$ is independently and identically distributed (i.i.d.) over time slots $t$ \footnote{This assumption is standard in the wireless networking literature and has been widely adopted in prior work (e.g.,~\cite{neely2010stochastic}).} with an \emph{unknown} mean $\mu_{l,n,m} \in [0,1]$. In practice, $\mu_{l,n,m}$ denotes the probability that a packet transmitted from user $n$ to AP $l$ at rate $r_m$ is successfully received and acknowledged. Possible physical-layer effects such as fading, coding, and inter-AP interference are implicitly captured in this average success probability, and users simultaneously served by the same AP are assumed to be interference-free.
Let $I_{l,n,m}(t)=1$ denote that user $n$ associated with AP $l$ transmits at rate $r_m$ in time slot $t$, and $I_{l,n,m}(t)=0$ otherwise. 
The green dots correspond to the \emph{decision variable} $I_{l,n,m}(t)$, as illustrated in Fig.~\ref{fig:system-architecture}.
Hence, the total received throughput of all users in frame $k$ is 
$R(kT)\triangleq\sum_{t=kT}^{(k+1)T-1}\sum_{l,n,m}S_{l,n}(kT)r_mX_{l,n,m}(t)I_{l,n,m}(t)$. Table~\ref{tab:parameters} summarizes these key parameters and their analogies in the MAB framework.

\begin{table}[h]
\centering
\caption{ Summary of Key Parameters and MAB Analogies.}
\renewcommand{\arraystretch}{1.5}
\footnotesize
\setlength{\tabcolsep}{3pt}
\begin{tabular}{|c|c|c|c|}
\hline
\textbf{Parameter} & \textbf{Description} & \textbf{Role} & \textbf{MAB Analogy} \\
\hline
$I_{l,n,m}(t)$ & Rate selection & Decision variable & Arm selection \\
\hline
$X_{l,n,m}(t)$ & Transmission result & Observed outcome & Stochastic reward \\
\hline
$\mu_{l,n,m}$ & Success probability & Unknown parameter & Expected reward \\
\hline
\end{tabular}
\label{tab:parameters}
\end{table}

Our goal is to maximize the cumulative expected throughput $\sum_{k=0}^{K-1}\bE[R(kT)]$ while ensuring fairness among users, indicating each user is scheduled for at least $\lambda_n\in(0,1)$ fraction of the time on average. If the statistics of throughput (i.e., $\mu_{l,n,m}, \forall l=1,2,\cdots,L, \forall n=1,2,\cdots,N, \forall m=1,2,\cdots,M$) are known, our objective can be achieved by selecting a randomized stationary schedule\footnote{The existence of such a randomized stationary policy can be shown by using the similar argument in \cite{neely2003dynamic} and its proof is omitted for brevity.} $\{q^{*}(\mb{S}), \forall \mb{S}\in\mc{S}\}$, 
where $q^*(\mb{S})$ is the probability of selecting a feasible schedule $\mb{S}$ and solves the following optimization problem: 
\vspace{-0.05 in}
\begin{align}
\max_{q(\mb{S})} &\quad \sum_{\mb{S}\in\mc{S}}q(\mb{S})\sum_{l=1}^{L}\sum_{n=1}^{N}S_{l,n}T\max_{m}r_m\mu_{l,n,m}\label{eqn:obj}\\
\text{s.t.} &\qquad \sum_{\mb{S}\in\mc{S}}q(\mb{S})\sum_{l=1}^{L}S_{l,n}\geq\lambda_n+\delta, \forall n=1,2,\ldots,N,\label{eqn:fairness}
\end{align} where the optimization is over all feasible
schedules in $\mc{S}$. Since the access constraints (e.g., each user associates with at
most one AP and each AP serves at most a certain number of
users per frame) are
already encoded in $\mc{S}$, they are automatically enforced
and do not need to appear as separate constraints.
The parameter $\delta>0$ acts as a tightness margin, ensuring that the constraints are strictly satisfied. This aligns with the notion of Slater's condition in convex optimization, which requires the existence of a strictly feasible solution.  The set of all per-user service-rate vectors achievable by such randomized stationary policies, namely
\begin{align}
\Lambda\triangleq\Big\{\Big(\sum_{\mathbf{S}\in\mathcal{S}}q(\mathbf{S})\sum_{l=1}^{L}S_{l,n}\Big)_{n=1}^{N}:\ q(\mathbf{S})\geq 0,\ \sum_{\mathbf{S}\in\mathcal{S}}q(\mathbf{S})=1\Big\}, \nonumber
\end{align}
forms the \emph{achievable region} of the network, which serves as an analogue of the capacity region in our setting. The slackness $\delta$ then measures how far $\bs{\lambda}=(\lambda_n)_{n=1}^{N}$ lies strictly inside $\Lambda$.  Here, $\max_{m}r_m\mu_{l,n,m}$ represents the maximum achievable throughput for user $n$ communicating with AP $l$ in each time slot, and thus $\sum_{l,n}S_{l,n}T\max_{m}r_m\mu_{l,n,m}$ represents its maximum throughput in one time frame if the schedule $\mb{S}=(S_{l,n})_{l,n}$ is selected. In the rest of the paper, we let $\mathbf{S}^{*}$ denote the feasible schedule selected by the optimal randomized stationary schedule $q^{*}(\mb{S})$.
Within each frame, each selected user $n$ associated with AP $l$ chooses $\mb{I}^{*}\triangleq(I_{l,n,m}^*)_{l,n,m}\in \argmax_{\mb{I}}\sum_{m=1}^{M}r_m\mu_{l,n,m}I_{l,n,m}$, i.e., selecting the rate with the maximum throughput, i.e., $\max_{m}r_m\mu_{l,n,m}$, in each slot. 

However, the throughput statistics are unknown \emph{a priori}
in practice. Therefore, each user must learn these statistics
(a.k.a.\ exploration) and select the empirically best
transmission rate so far (a.k.a.\ exploitation). This process
inevitably leads to a throughput loss compared to the scenario
where the throughput statistics are known in advance.
To quantify this throughput loss, we adopt the
notion of \emph{cumulative regret}, defined as the gap
between the expected accumulated throughput
under the optimal policy with known throughput statistics
and that achieved by our algorithm, i.e.,
\begin{align*}
\text{Reg}&(KT)\triangleq\underbrace{\sum_{k,l,n}\bE\left[S_{l,n}^*\sum_{t=kT}^{(k+1)T-1}\sum_{m=1}^{M}r_m\mu_{l,n,m}I_{l,n,m}^*\right]}_{\triangleq\,\mathrm{OPT}(KT)}\nonumber\\
&-\underbrace{\sum_{k,l,n}\bE\left[S_{l,n}(kT)\sum_{t=kT}^{(k+1)T-1}\sum_{m=1}^{M}r_m\mu_{l,n,m}I_{l,n,m}(t)\right]}_{\triangleq\,\mathrm{ALG}(KT)}.
\end{align*}
Since $\mathrm{OPT}(KT)$ is a fixed constant
determined by the throughput statistics (i.e.,
$\mu_{l,n,m}$), minimizing the cumulative regret is
equivalent to maximizing the expected accumulated throughput $\mathrm{ALG}(KT)$
achieved by our algorithm. Our goal is to design a joint
user scheduling and rate adaptation algorithm that not only
meets the desired fairness requirement but also minimizes the
cumulative regret over consecutive $K$ time frames.

We note that our framework is
also relevant to latency-sensitive uplink applications. Since
the cumulative throughput measures the total data delivered
within $KT$ time slots of fixed duration, configuring a
shorter time slot duration effectively imposes stricter
per-slot delivery deadlines. In this way, the throughput
maximization objective naturally translates into a measure of
timely delivery performance.

\section{Algorithm Design and Performance Analysis}
\label{sec:alg}
In this section, we develop an online-learning-based joint user scheduling and rate adaptation algorithm by integrating the key idea of the well-known UCB algorithm and virtual queue techniques while respecting the different time scales of user scheduling and rate adaptation. In particular, in the time scale for user schedule, the virtual queues are introduced to guarantee the desired fairness constraint (see \cite{neely2010stochastic} for an overview). In contrast, in the time scale for rate adaptation, the UCB approach is utilized to deal with the fundamental exploitation-exploration tradeoff in online learning for each user to identify the best transmission rate while achieving a minimum cumulative regret.

To ensure fairness among users, we maintain a virtual queue
for each user that tracks how much scheduling
time the user is still owed relative to its fairness
requirement. Specifically, each user $n$ is
entitled to be scheduled for at least a $\lambda_n$ fraction
of time on average. If in a given frame user $n$ is not
scheduled (i.e., $\sum_{l=1}^{L}S_{l,n}(kT)=0$), its virtual
queue-length increases, reflecting that the system has
under-served this user. Conversely, when user $n$ is scheduled
(i.e., $\sum_{l=1}^{L}S_{l,n}(kT)=1$), the virtual
queue-length decreases. A larger virtual queue-length thus
signals a more urgent need to schedule that user in subsequent
frames. In particular, let $Q_n(kT)$ be the virtual
queue-length of user $n$ at the beginning of time frame $k$,
and its evolution over time frames is described as follows:
\begin{align}
\label{eqn:virtualQ}
Q_n((k+1)T)=\Big(Q_n(kT)+\lambda_n-\sum_{l=1}^{L}S_{l,n}(kT)+\epsilon_k\Big)^{+},
\end{align}
for $k=0,1,2,\ldots$, where $(x)^{+}\triangleq\max\{x,0\}$ and $\epsilon_k>0$ is some control parameter that will be specified later. We set $Q_n(0)=0$ as the system starts at $k=0$. 

Let $H_{l,n,m}(t)$ be the number of time slots that user $n$ is associated with AP $l$ and transmits at rate $r_m$ until time slot $t$, i.e., $H_{l,n,m}(t)\triangleq\sum_{\tau=0}^{t-1}S_{l,n}(\lfloor\tau/T\rfloor T)I_{l,n,m}(\tau)$, where $\lfloor x\rfloor$ denotes the maximum integer that is not greater than $x$. We set $H_{l,n,m}(0)=0$ due to the fact that the system starts at $t=0$. We use $\ol{\mu}_{l,n,m}(t)$ to denote the fraction of successful transmissions when user $n$ is associated with AP $l$ and transmits at rate $r_m$ until time slot $t$, i.e.,
\begin{align}
\ol{\mu}_{l,n,m}(t)\triangleq\frac{\sum_{\tau=0}^{t-1}S_{l,n}(\lfloor \tau/T\rfloor T)X_{l,n,m}(t)I_{l,n,m}(t)}{H_{l,n,m}(t)}. \nonumber
\end{align}
If $H_{l,n,m}(t)=0$, we set $\ol{\mu}_{l,n,m}(t)=1$. Let $w_{l,n,m}(t)$ denote the UCB estimate of user $n$ associated with AP $l$ using rate $r_m$ in time slot $t$, which can be defined below:
\begin{align}
\label{eqn:UCBweight}
 w_{l,n,m}(t)\triangleq \min\left\{\ol{\mu}_{l,n,m}(t)+\sqrt{\frac{3\log t}{2H_{l,n,m}(t)}},1\right\},
\end{align}
where $\sqrt{3\log t/(2H_{l,n,m}(t))}$ is the exploration bonus term that measures the uncertainty of the sample mean $\ol{\mu}_{l,n,m}(t)$. Note that a smaller $H_{l,n,m}(t)$ implies less exploration on user $n$ using rate $r_m$ and thus more inaccuracy in the estimate $\ol{\mu}_{l,n,m}(t)$, in which case user $n$ is encouraged to transmit at rate $r_m$ for further exploration. In \eqref{eqn:UCBweight}, we use the truncated version of the UCB estimate, since the successful transmission probability is at most $1$. When $H_{l,n,m}(t)=0$, we set $w_{l,n,m}(t)=1$, i.e., if user $n$ has not transmitted at rate $r_m$ until time slot $t$, it should have the highest priority to be served.

\begin{algorithm}[H]
\caption{\underline{O}nline-Learning-based Joint \underline{U}ser Associa\underline{t}ion and Scheduling and Rate Adap\underline{ta}tion (OUTTA) Algorithm}
At the beginning of frame $k$, select a feasible schedule $\wh{\mb{S}}(kT)\triangleq(\wh{S}_{l,n}(kT),\forall l,\forall n)$ satisfying
\begin{align*}
\wh{\mb{S}}(kT)\in 
&\argmax_{\mb{S}\in\mc{S}} \sum_{l,n}S_{l,n}\bigg(Q_n(kT) \nonumber\\
&\qquad\qquad + \eta_k T \max_{m}r_mw_{l,n,m}(kT)\bigg),
\end{align*}
where $\eta_k=\delta\sqrt{k}/T$ (with $\eta_0=\delta/(2T)$). Then, update the virtual queue-lengths according to \eqref{eqn:virtualQ} with $\epsilon_k=(4r_MLN^{1.5}+1)/(2\sqrt{k+1})$.

Within each time slot $t$ in frame $k$, i.e., $t=kT,kT+1,\ldots,(k+1)T-1$, each selected user $n$ associated with AP $l$ (i.e., $\wh{S}_{l,n}(kT)=1$) chooses the rate index $\wh{m}_n(t)$ (i.e., $\wh{I}_{l,n,\wh{m}_n(t)}(t)=1$ and $\wh{I}_{l,n,m}(t)=0, \forall m\neq\wh{m}_n(t)$) such that 
\begin{align*}
    \wh{m}_n(t)\in\argmax_{m}r_mw_{l,n,m}(t).
\end{align*}
\label{alg:OL}
\vspace{-0.1 in}
\end{algorithm}

On one hand, we would like to schedule users with large virtual queue-lengths in each time frame to meet the desired fairness constraint. On the other hand, in order to achieve a low cumulative regret, we prefer to schedule users and select their rates with large UCB weights in each time slot. This motivates the following online-learning-based joint user scheduling and rate adaptation algorithm, as shown in Algorithm \ref{alg:OL}.

In the proposed \hc{} algorithm, the increasing sequence $\{\eta_k\}_{k\geq0}$ balances the virtual queue-lengths and the UCB estimates for throughput statistics over time frames. Initially, the \hc{} algorithm puts a larger weight on the virtual queue-lengths to quickly guarantee desired fairness while learning the best transmission rate for each user, and then emphasizes more on the UCB weight to ensure a smaller cumulative regret. The parameter $\eta_k$ requires the exact knowledge of the slackness constant, which is usually unavailable in practice. We will demonstrate that the \hc{} algorithm with inaccurate slackness constants still performs well via simulations in Section \ref{sec:sim}. Different from prior works on combinatorial bandits with fairness constraints (e.g., \cite{liu2021efficient}), the user scheduling and rate selection have different time scales. This requires carefully manipulating the virtual queue-lengths and UCB weights and decoupling them in an appropriate way in the performance analysis.

Next, we characterize the cumulative fairness violation of the proposed \hc{}{} algorithm. 
\begin{proposition} [Cumulative Fairness Violation]
\label{prop:violation}
Under the \hc{} algorithm, if $\exists k'$, such that for any $k \geq k'$, $\epsilon_k \leq \delta/2$, the cumulative fairness violation over $K$ time frames can be upper bounded below: 
\begin{align*}
&\sum_{n=1}^{N}\left(\bE\left[\sum_{k=0}^{K-1}\sum_{t=k T}^{(k+1)T-1}\left(\lambda_n-S_n(t)\right)\right]\right)^{+} \nonumber\\
&\leq NT\left(g(N,\delta,r_M)-\sqrt{K}\right)^{+},
\end{align*}
where $g(N,\delta,r_M)=\frac{74N^{2.5}}{\delta}\log\left(\frac{18N}{\delta}\right)+(4r_ML+3)N^{1.5}+\frac{N^{1.5}(6+\delta^2)}{\delta}+1+\frac{N^{1.5}(4r_MLN^{1.5}+1)^2}{\delta}\left(\frac{1}{\delta}+1\right)$ is a constant depending on system parameters such as $N,\delta$ and $r_M$. 
\end{proposition}

\begin{IEEEproof}
We first select the Lyapunov function 
\begin{align*}
    V(kT)\triangleq\|\mb{Q}(kT)\|,
\end{align*}
and prove that the Lyapunov function has an expected negative drift when $V(kT)$ is sufficiently large and its drift is absolutely bounded. Then according to \cite[Lemma 11]{liu2021efficient}, $\bE\left[\|\mb{Q}(kT)\|_1\right]$ can be upper bounded. Finally, we can derive the upper bound of the cumulative fairness violation by combining the dynamics of virtual queue-lengths and the analysis of $\bE\left[\|\mb{Q}(kT)\|_1\right]$. Please see Appendix \ref{APP:prop:violation} for the detailed proof.
\end{IEEEproof}
\begin{remarks}
Note that $k'$ always exists due to the fact that $\{\epsilon_k\}_{k\geq0}$ is an decreasing sequence. 
Given the condition, we can also see from Proposition \ref{prop:violation} that the \hc{} algorithm achieves zero cumulative fairness violation when $K\geq g^2(N,\delta,r_M)$. Moreover, the number of frames required for achieving zero cumulative fairness violation is independent of frame size $T$ and thus the required number of time slots for achieving zero cumulative fairness violation linearly increases with the frame size $T$. Furthermore, the amount of cumulative fairness violation linearly increases with the frame size $T$. All these observations will be demonstrated via simulations in Section \ref{sec:sim}.
\end{remarks}

We derive an upper bound on the cumulative regret under the \hc{} algorithm.

\begin{proposition} [Cumulative Regret] 
\label{prop:regret}
Under the \hc{} algorithm with $\epsilon_k \leq \delta$, the cumulative regret $\text{Reg}(KT)$ over $K$ time frames can be upper bounded as follows:
\begin{align*}
&\text{Reg}(KT)\leq \frac{Nr_{M}T(4r_MLN^{1.5}+1)^2}{4\delta^2}+ 2\sqrt{K}NT(\delta+\frac{3}{2\delta})  \nonumber\\
&+ \frac{NT(2\delta+1)^2(4r_MLN^{1.5}+1)^3}{16\delta^4} \nonumber\\
&+ LMNr_M\left(T+3+\frac{5\pi^2}{6}+(T+1)\log(KT)\right)\nonumber\\
&+(T+4)r_M\sqrt{6LMNS_{\max}KT\log(KT)}\nonumber\\
&+r_M\sqrt{\frac{3LMNS_{\max}KT}{2\log T}}\nonumber\\
 &=O\bigg(NT\sqrt{K}+LMNT\log(KT)  \nonumber\\ &\qquad\qquad\qquad\qquad\qquad+T\sqrt{LMNKT\log(KT)}\bigg).
\end{align*}
\end{proposition}
\begin{IEEEproof}
We perform the drift-plus-penalty analysis.
Unlike prior work on the regret analysis (e.g., \cite{liu2021efficient}), the different time scales of user scheduling and rate selection impose unique challenges on the corresponding regret bound analysis. In particular, we need to upper bound the regret by carefully decoupling the user decision and rate adaption in different time scales. Please see Appendix \ref{APP:prop:regret} for the detailed proof.
\end{IEEEproof}
\begin{remarks}
For the impact of the number of frames $K$ on the regret performance, our derived regret upper bound has the same order $O(\sqrt{K\log K})$ as the instance-independent upper bound for the classical UCB algorithm. While the derived regret upper bound increases with the frame size $T$, the simulations demonstrate that the frame size has a marginal impact on the regret performance. The reason is that each user has sufficient time to identify its best transmission rate under different frame sizes.
\end{remarks}

\begin{remarks}
\label{rmk:generality}
Since the wireless technology enters the analysis only through the feasible set $\mathcal{S}$ and the success probabilities $\mu_{l,n,m}$, the Lyapunov-drift argument and the resulting guarantees apply unchanged across diverse MAC/PHY layers.
\end{remarks}

\begin{discussion}
\label{disc:approx}
Proposition~\ref{prop:regret} can be restated as an approximation ratio with respect to the optimal value of \eqref{eqn:obj} and \eqref{eqn:fairness}. Let $\rho(K)\triangleq\mathrm{ALG}(KT)/\mathrm{OPT}(KT)=1-\mathrm{Reg}(KT)/\mathrm{OPT}(KT)$ denote the fraction of that optimal value attained by the \hc{} algorithm over $K$ time frames. Every time frame serves $L$ users at a strictly positive expected rate, so $\mathrm{OPT}(KT)=\Theta(K)$ grows linearly in the number of frames, while Proposition~\ref{prop:regret} gives $\mathrm{Reg}(KT)=O(\sqrt{K\log K})$. Therefore $\rho(K)\geq1-O(\sqrt{\log K/K})\rightarrow1$ as $K\rightarrow\infty$, that is, \hc{} is a $\rho(K)$-approximation over any finite horizon, and its gap $1-\rho(K)=\mathrm{Reg}(KT)/\mathrm{OPT}(KT)$ is determined directly by the regret bound.

It is worth being explicit about why this ratio is horizon-dependent, whereas the approximation and competitive ratios of~\cite{tan2013fast,lin2005impact,lin2013dynamic,chen2013markov} are constants. In those works the quantity being bounded is the suboptimality of a decision rule evaluated against a known objective, so the ratio is determined once and does not change with the length of the operating period. In our setting the gap between $\mathrm{ALG}(KT)$ and $\mathrm{OPT}(KT)$ arises only because $\mu_{l,n,m}$ is unknown and has to be estimated by the UCB estimate $w_{l,n,m}(t)$, which is the exploration-exploitation trade-off. Since the exploration bonus $\sqrt{3\log t/(2H_{l,n,m}(t))}$ decays as $H_{l,n,m}(t)$ grows, the accumulated shortfall is $O(\sqrt{K\log K})$ rather than linear in $K$, so the ratio improves with the horizon instead of being fixed. No horizon-free constant can describe this behavior, because the gap is by construction a function of how long the algorithm has been learning.

The guarantee is moreover instance-independent, since Proposition~\ref{prop:regret} is of the same order as the instance-independent bound for the classical UCB algorithm, so the rate holds for every instance and carries no dependence on the gaps. Its numerical value at a given horizon is nevertheless instance-specific, since it involves $\mathrm{OPT}(KT)$ and hence the unknown $\mu_{l,n,m}$, and the explicit constants in Proposition~\ref{prop:regret} are worst-case. We therefore state the guarantee as a rate and report the realized ratio on real-world traces in Section~\ref{sec:sim}.

\end{discussion}

\section{low-complexity algorithm design}
The \hc{} algorithm schedules the subset of users with the maximum total weight of virtual queue-length and the UCB estimation in each time frame. Note that $|\mc{S}|$ can grow exponentially with the number of users due to the existence of interference constraints. Since the \hc{} algorithm necessitates evaluating all feasible schedulers in $\mc{S}$ to schedule the subset of users with the maximum total weight in each frame, it results in a high computational complexity. To address this, we propose an alternative algorithm that reduces the number of comparison steps in each frame, thereby lowering the computational complexity of the \hc{} algorithm, as shown in Algorithm \ref{alg:PC-OL}. 
\begin{algorithm}[H] 
\caption{\underline{P}ick and \underline{C}ompare \underline{O}nline-Learning-based Joint \underline{U}ser Associa\underline{t}ion and Scheduling and Rate Adap\underline{ta}tion (\lc{}) Algorithm}
At the beginning of frame $k$, select $\mb{R}(kT) \in \mathcal{S}$ uniformly at random and compare $\mb{R}(kT)$ with the previous schedule $\wh{\mb{S}}^{\text{PC}}((k-1)T)\triangleq(\wh{S}_{l,n}^{\text{PC}}((k-1)T),\forall l,\forall n)$ satisfying
\begin{align*}
\wh{\mb{S}}^{\text{PC}}(kT)\in 
&\argmax_{\mb{S}\in \{\mb{R}(kT),\wh{\mb{S}}^{\text{PC}}((k-1)T)\}} \sum_{l,n}S_{l,n}\bigg(Q_n(kT) \nonumber\\
&\qquad\qquad + \eta_k T \max_{m}r_mw_{l,n,m}(kT)\bigg),
\end{align*}
where $\eta_k=\delta\sqrt{k}/T$ (with $\eta_0=\delta/(2T)$). Then, update the virtual queue-lengths according to \eqref{eqn:virtualQ} with $\epsilon_k=(4r_MLN^{1.5}+1)/(2\sqrt{k+1})$.

Within each time slot $t$ in frame $k$, i.e., $t=kT,kT+1,\ldots,(k+1)T-1$, each selected user $n$ associated with AP $l$ (i.e., $\wh{S}_{l,n}^{\text{PC}}(kT)=1$) chooses the rate index $\wh{m}_n(t)$ (i.e., $\wh{I}_{l,n,\wh{m}_n(t)}(t)=1$ and $\wh{I}_{l,n,m}(t)=0, \forall m\neq\wh{m}_n(t)$) such that 
\begin{align*}
    \wh{m}_n(t)\in\argmax_{m}r_mw_{l,n,m}(t).
\end{align*}
\label{alg:PC-OL}
\end{algorithm}

We leverage the idea of the PC algorithm and develop the following \lc{} algorithm that first randomly picks one feasible schedule in each frame and then selects the feasible schedule with the maximum weight among the randomly picked feasible schedule and the selected schedule in the previous frame.
We use $\mc{R}(kT)$ to denote the feasible schedule randomly picked in frame $k$. This PC design cuts down the number of comparison steps from $|\mc{S}|$ in the \hc{} algorithm to one in the \lc{} algorithm. The rest of \lc{} design is similar to \hc{}. To sum up, our \lc{} algorithm design decreases the switching overhead incurred by frequent user handoffs between APs by employing a frame-based design for scheduling and association. Additionally, it adopts the pick-and-compare strategy to circumvent evaluating every feasible schedule, as \hc{} does, thereby substantially lowering the computational complexity of the MaxWeight design. The trade-off between computational efficiency and performance has been theoretically analyzed in our recent work \cite{wu2025low}, where we rigorously studied the cumulative regret, fairness guarantees, and complexity of the PC-based design. We evaluate the performance of \lc{} via simulations based on real-world data in Section \ref{sec:sim}.

\section{Simulations}
\label{sec:sim}
In this section, we evaluate the performance of our proposed \hc{} algorithm and \lc{} algorithm via simulations based on real-world data.

\noindent\textbf{Experimental Setup:}
We consider a 60 GHz mmWave short-range communication network (e.g., IEEE 802.11ad \cite{IEEE802p11ad}) deployed in a classroom environment, where three APs are positioned and 10 user devices are uniformly distributed across the classroom, as shown in Fig.~\ref{fig:distribution}.
Ideally, commodity off-the-shelf (COTS) 802.11ad devices would be used to evaluate our algorithm. However, existing COTS 802.11ad routers do not provide sufficient control over rate adaptation, making them unsuitable for our experimental requirements.

To address this limitation, we developed a 60 GHz mmWave testbed to collect end-to-end channel quality traces and conduct trace-driven simulations based on real-world measurements, as illustrated in Fig.~\ref{fig:mmWave_exp}. In principle, all 30 AP–user links should be measured simultaneously. However, the available hardware does not support concurrent measurements for all links. Therefore, we sequentially measured the 30 links one at a time, which is a common practice known as single-sounder sequential measurement~\cite{imoize2021standard}.

The testbed consists of a transmitter located at an AP and a receiver located at a user device. Both the transmitter and the receiver are built from a computer for baseband signal processing, a USRP X310 for signal shaping, and Sivers EVK06002 for  up/down frequency conversion.
Simplified IEEE 802.11ad PHY-layer signal processing modules are implemented at both ends to enable real-time 802.11ad OFDM data packet transmission from the AP to the user device. The receiver demodulates the OFDM signal and records the post–signal-to-noise ratio (postSNR) for each packet. Fig.~\ref{fig:receiver} shows the mmWave receiver diagram, which includes the spectrum display, demodulated signal constellation, and decoded video stream. The post-SNR is computed directly from the demodulated signal constellation.

\begin{figure}
    \centering
    \includegraphics[width=0.25\textwidth, height=0.15\textwidth,keepaspectratio]{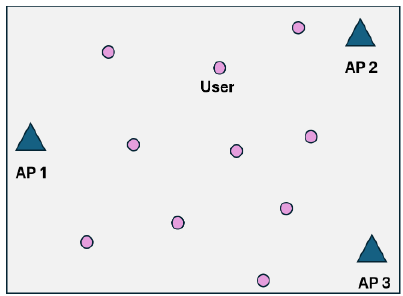}
\caption{The distribution of 3 APs and 10 users.}
\label{fig:distribution}
\end{figure}

\begin{figure}
\centering
\begin{minipage}[t]{2.7in}
    \centering
    \includegraphics[width=0.9\textwidth, height=0.1\textheight,keepaspectratio]{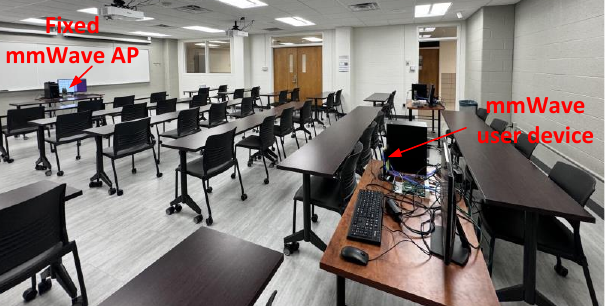}
    \label{fig:exp_setup_annotated}
\end{minipage}
\begin{minipage}[t]{0.7in}    
    \centering
    \includegraphics[width=0.5\textwidth, height=0.1\textheight,keepaspectratio]{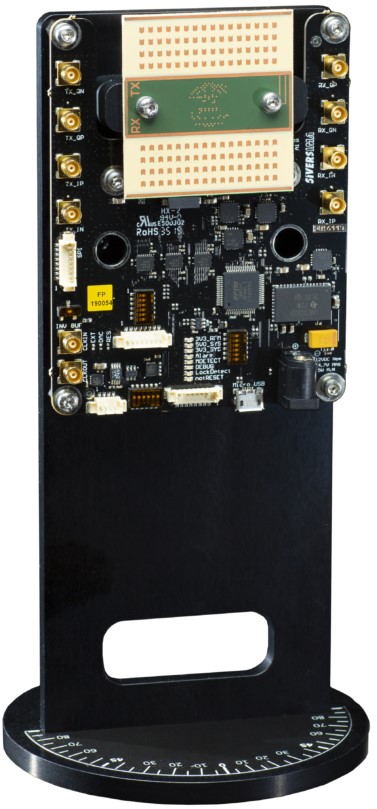}
    \label{fig:mmWave_AP_annotated}
\end{minipage}
\caption{Left: Experimental setup in a classroom. Right: A 60GHz mmWave transceiver with phased-array antenna.}
\label{fig:mmWave_exp}
\end{figure}

\begin{figure}
    \centering
    \includegraphics[width=0.5\textwidth, height=0.2\textwidth,keepaspectratio]{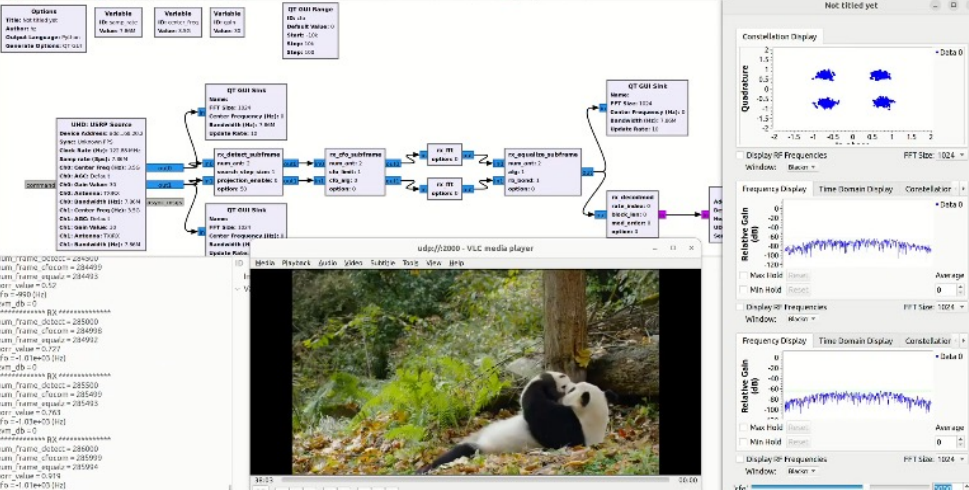}
\caption{The receiver's diagram.}
\label{fig:receiver}
\end{figure}

\begin{table}[]
\centering
\scriptsize
\setlength{\tabcolsep}{2.5pt}
\caption{EVM table specified in IEEE 802.11ad standard \cite{IEEE802p11ad}
(B: BPSK; Q: QPSK; 16Q: 16-QAM).}
\label{tab:rate_adaptation}
\begin{tabular}{|c|c|c|c|c|c|c|c|c|c|}
\hline
index ($m$) & 1 & 2 & 3 & 4 & 5 & 6 & 7 & 8 & 9
\\
\hline
\!\!\!\!\!postSNR (dB)\!\!\!\!\! & \!\!\! -7   \!\!\! & \!\!\! -9    \!\!\! & \!\!\! -10        \!\!\! & \!\!\! -11         \!\!\! & \!\!\! -12   \!\!\! & \!\!\! -14   \!\!\! & \!\!\! -15       \!\!\! & \!\!\! -16  & \!\!\! -17 \\ \hline
\!\!\!\!\!Modulation\!\!\!\!\! & \!\!\! B \!\!\! & \!\!\! B  \!\!\! & \!\!\! Q \!\!\! & \!\!\! Q  \!\!\! & \!\!\! Q  \!\!\! & \!\!\! Q  \!\!\! & \!\!\! 16Q \!\!\! & \!\!\! 16Q & \!\!\! 16Q \\ \hline
\!\!\!\!\!Coding rate\!\!\!\!\! & \!\!\! 1/2  \!\!\! & \!\!\! 5/8   \!\!\! & \!\!\! 1/2   \!\!\! & \!\!\! 5/8  \!\!\! & \!\!\! 3/4   \!\!\! & \!\!\! 13/16 \!\!\! & \!\!\!  1/2  \!\!\! & \!\!\! 5/8 & \!\!\! 3/4 \\ \hline
\!\!\!\!\!$\gamma$ (postSNR)\!\!\!\!\! & \!\!\! 0.5  \!\!\! & \!\!\! 0.63 \!\!\! & \!\!\! 1    \!\!\! & \!\!\! 1.25 \!\!\! & \!\!\! 1.5 \!\!\! & \!\!\! 1.63 \!\!\! & \!\!\! 2   \!\!\! & \!\!\! 2.5 & \!\!\! 3 \\ \hline
\!\!\!\!\!Rate (Gbps)\!\!\!\!\! & \!\!\! 0.73  \!\!\! & \!\!\! 0.91 \!\!\! & \!\!\! 1.46    \!\!\! & \!\!\! 1.825 \!\!\! & \!\!\! 2.19 \!\!\! & \!\!\! 2.37 \!\!\! & \!\!\! 2.92   \!\!\! & \!\!\! 3.65 & \!\!\! 4.38 \\ \hline
\end{tabular}
\end{table}

\noindent\textbf{Data Collection:}
Prior to data collection for each link, we first adjust the beam indices at both the mmWave transmitter and receiver to align their directions. This procedure emulates the beam search protocol implemented in 5G and IEEE 802.11ad/ay systems. Once beam alignment is established, we begin recording the postSNR at the receiver. During data collection, both the transmitter and receiver remain stationary with an unobstructed line-of-sight (LOS) path, although human activity occurs in the surrounding environment.
For each link, the receiver records the postSNR of decoded signal constellations every 1 ms. The resulting dataset forms a $30\times 30000$ matrix, where each entry represents the instantaneous end-to-end channel quality (i.e., postSNR) from a given AP to a user device. We will release this dataset publicly to facilitate future research.
A representative demonstration of mmWave physical-layer real-time video transmission from our group's testbed is available online~\cite{mmwavedemo}.

\noindent\textbf{Interpretation of PostSNR:} In communication engineering, postSNR is also referred to as the error vector magnitude (EVM). It serves as a comprehensive indicator of instantaneous link quality. Based on the measured postSNR, the achievable data rate of an IEEE~802.11ad link is computed as
$
r(\mathrm{postSNR}) = f \cdot \frac{\tau_{ofdm}}{\tau_{gi} + \tau_{ofdm}} \cdot \frac{N_{data}}{N_{fft}} \cdot \gamma(\mathrm{postSNR})$,
where $f = 2.64~\mathrm{GHz}$ is the sampling rate, 
$\tau_{gi} = 36.36~\mathrm{ns}$ is the guard interval duration, 
$\tau_{ofdm} = 194.56~\mathrm{ns}$ is the OFDM symbol duration, 
$N_{data} = 336$ is the number of data subcarriers, 
$N_{fft} = 512$ is the FFT size, 
and $\gamma(\mathrm{postSNR})$ is given by the second-to-last row of Table~\ref{tab:rate_adaptation} for each value.
For example, if the measured postSNR of a link is $-13~\mathrm{dB}$, Table~\ref{tab:rate_adaptation} specifies that the transmitter should select MCS index~5 for packet transmission, resulting in a data rate of $2.19~\mathrm{Gbps}$. Instead, if MCS index~8 is selected, the receiver would be unable to successfully decode the packet. This MCS look-up method is widely adopted in system-level simulations within the 5G industry. Table~\ref{tab:rate_adaptation} can thus be regarded as a simplified abstraction of the PHY-layer of the communication system.

\noindent
\textbf{Simulation Setup:}
We consider $N=10$ users and assume $L=3$ APs, where each user can associate at most one AP and each AP can schedule multiple users in each time frame. Here, we set the maximum number of users each AP can simultaneously serve to one as a representative example. We set each user's desired scheduling fraction as $\bs{\lambda}=\frac{2.1}{55}\times[1,2,3,4,5,6,7,8,9,10]$.
There are nine rates available for selection, which are shown in the last row of Table~\ref{tab:rate_adaptation}.
The data packet transmission result is estimated as follows. 
Suppose the AP selects $r_m$ for a user device in the scheduling phase and the user device measures $snr$ as its postSNR, if $r(snr) \ge r_m$, the data packet transmission is successful; otherwise, it fails. Since the exact value of the slackness constant $\delta$ is unknown in this system, we conduct a robustness evaluation by varying $\delta$ across a representative range: $\delta = [0.05,0.1,0.15]$. This setup simulates scenarios where the algorithm operates under underestimated, exact, or overestimated slackness assumptions.

\begin{figure*}[!htbp]
\centering 
\vspace{-0.3in}
\subfloat[User Scheduling Fraction]{
\label{fig:trace:fraction}
\includegraphics[width=0.33\textwidth, height=0.2\textheight,keepaspectratio]{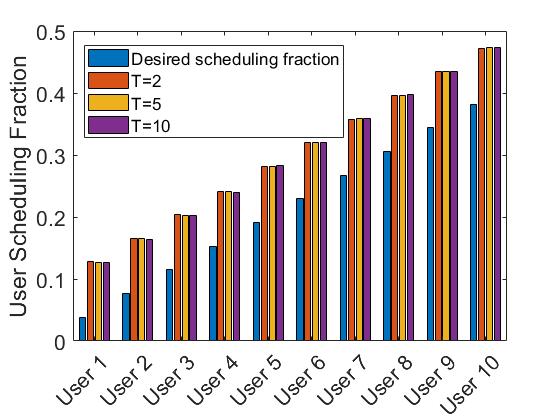}
\hspace{-0.2in}
} 
\subfloat[Cumulative Fairness Violation]{ \label{fig:trace:fairness}
\includegraphics[width=0.33\textwidth, height=0.2\textheight,keepaspectratio]{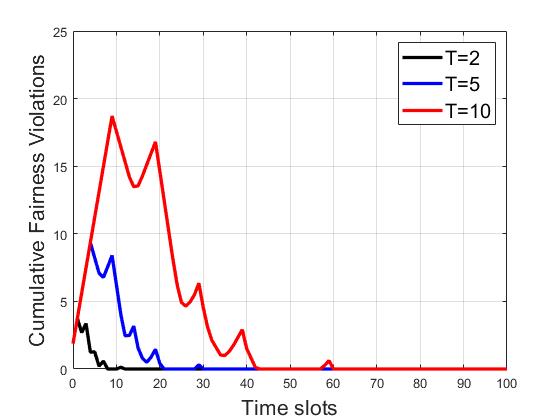}
\hspace{-0.2in}} 
\subfloat[Cumulative Regret]{
\label{fig:trace:regret}
\includegraphics[width=0.33\textwidth, height=0.2\textheight,keepaspectratio]{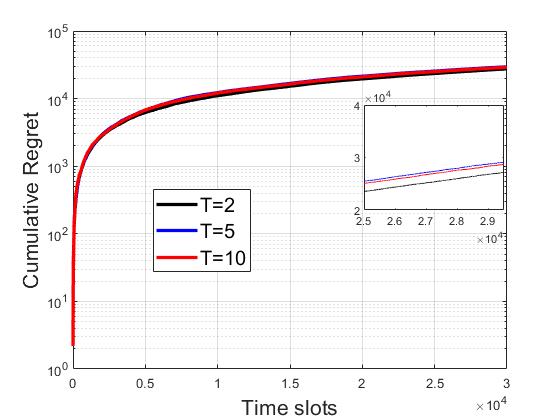}
} 
\caption{Impact of frame size $T$ for \hc{} in Trace-based simulation.}
\label{fig:trace}
\end{figure*}

\begin{figure*}[!htbp]
\centering 
\vspace{-0.3in}
\subfloat[User Scheduling Fraction]{
\label{fig:trace:deltafraction}
\includegraphics[width=0.33\textwidth, height=0.2\textheight,keepaspectratio]{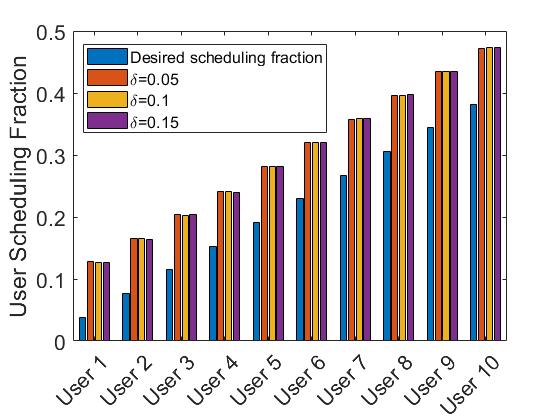}
\hspace{-0.2in}
} 
\subfloat[Cumulative Fairness Violation]{ \label{fig:trace:deltafairness}
\includegraphics[width=0.33\textwidth, height=0.2\textheight,keepaspectratio]{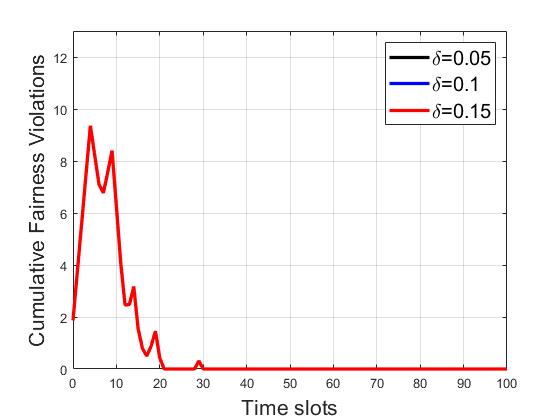}
\hspace{-0.2in}} 
\subfloat[Cumulative Regret]{
\label{fig:trace:deltaregret}
\includegraphics[width=0.33\textwidth, height=0.2\textheight,keepaspectratio]{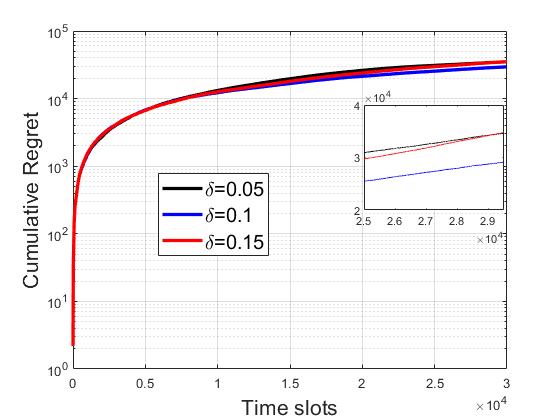}
} 
\caption{Impact of slackness constant $\delta$ for \hc{} in Trace-based simulation.}
\label{fig:deltatrace}
\end{figure*}

\begin{figure*}[!htbp]
\centering 
\vspace{-0.3in}
\subfloat[User Scheduling Fraction]{
\label{fig:pac_trace:fraction}
\includegraphics[width=0.33\textwidth, height=0.2\textheight,keepaspectratio]{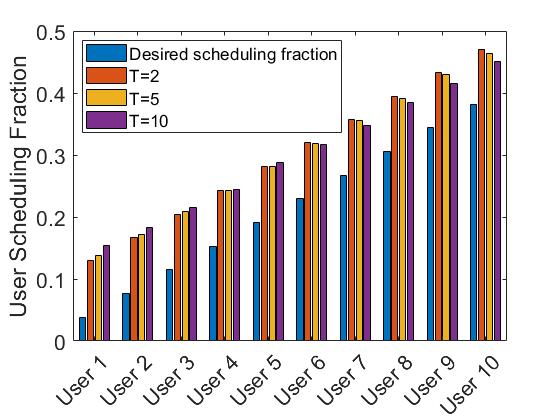}
\hspace{-0.2in}
} 
\subfloat[Cumulative Fairness Violation]{ \label{fig:pac_trace:fairness}
\includegraphics[width=0.33\textwidth, height=0.2\textheight,keepaspectratio]{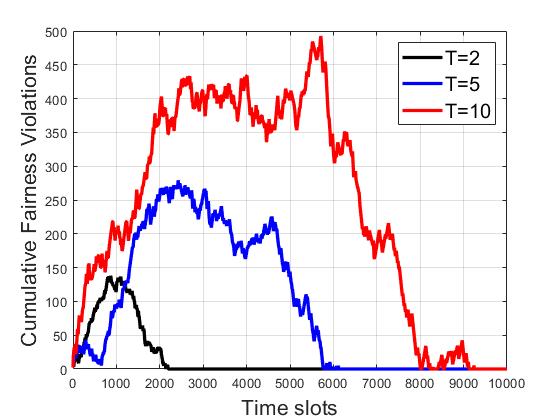}
\hspace{-0.2in}} 
\subfloat[Cumulative Regret]{
\label{fig:pac_trace:regret}
\includegraphics[width=0.33\textwidth, height=0.2\textheight,keepaspectratio]{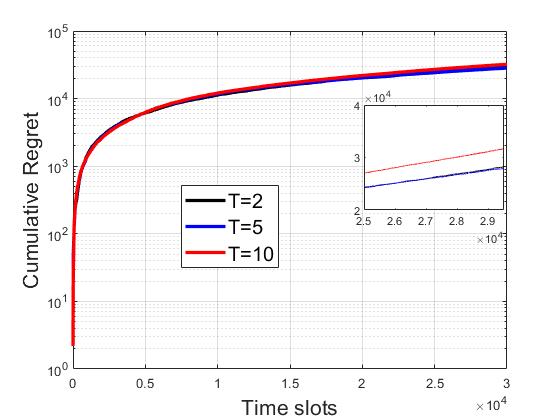}
} 
\caption{Impact of frame size $T$ for \lc{} in Trace-based simulation.}
\label{fig:pac_trace}
\end{figure*}

\begin{figure*}[!htbp]
\centering 
\vspace{-0.3in}
\subfloat[User Scheduling Fraction]{
\label{fig:pac_trace:fractionDelta}
\includegraphics[width=0.33\textwidth, height=0.2\textheight,keepaspectratio]{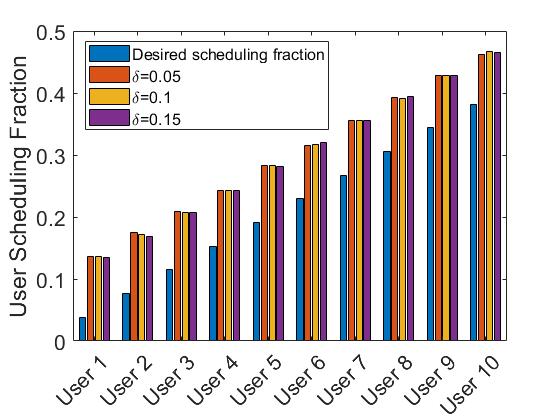}
\hspace{-0.2in}
} 
\subfloat[Cumulative Fairness Violation]{ \label{fig:pac_trace:fairnessDelta}
\includegraphics[width=0.33\textwidth, height=0.2\textheight,keepaspectratio]{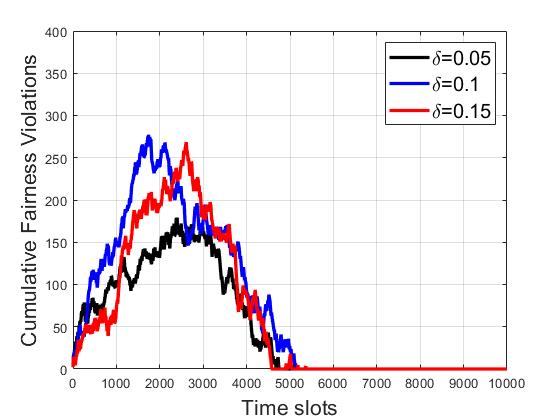}
\hspace{-0.2in}
} 
\subfloat[Cumulative Regret]{
\label{fig:pac_trace:regretDelta}
\includegraphics[width=0.33\textwidth, height=0.2\textheight,keepaspectratio]{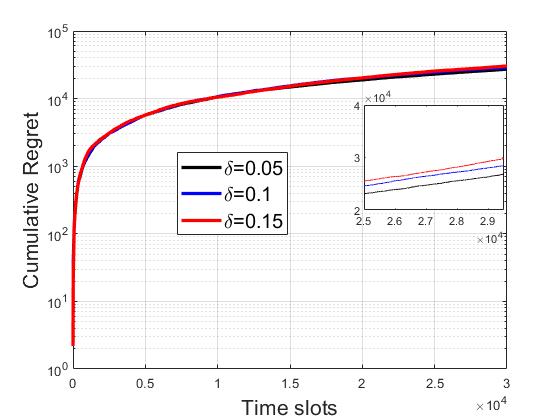}
} 
\caption{Impact of slackness constant $\delta$ for \lc{} in Trace-based simulation.}
\label{fig:pac_trace:delta}
\end{figure*}

\begin{figure*}[!htbp]
\centering
\vspace{-0.3in}
\subfloat[\hc{}]{
\label{fig:approx:outta}
\includegraphics[width=0.42\textwidth, height=0.22\textheight,keepaspectratio]{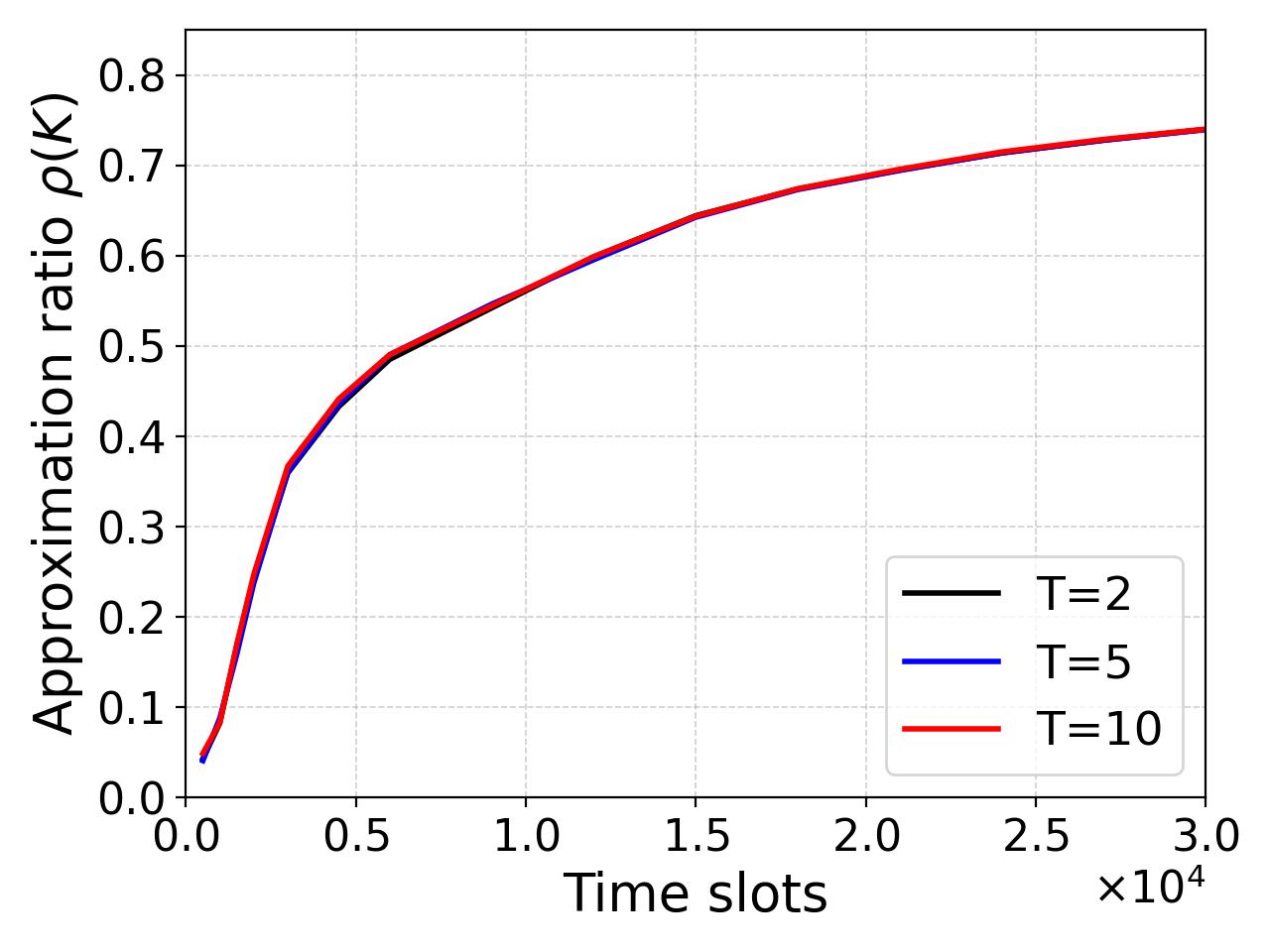}
\hspace{0.1in}
}
\subfloat[\lc{}]{
\label{fig:approx:pac}
\includegraphics[width=0.42\textwidth, height=0.22\textheight,keepaspectratio]{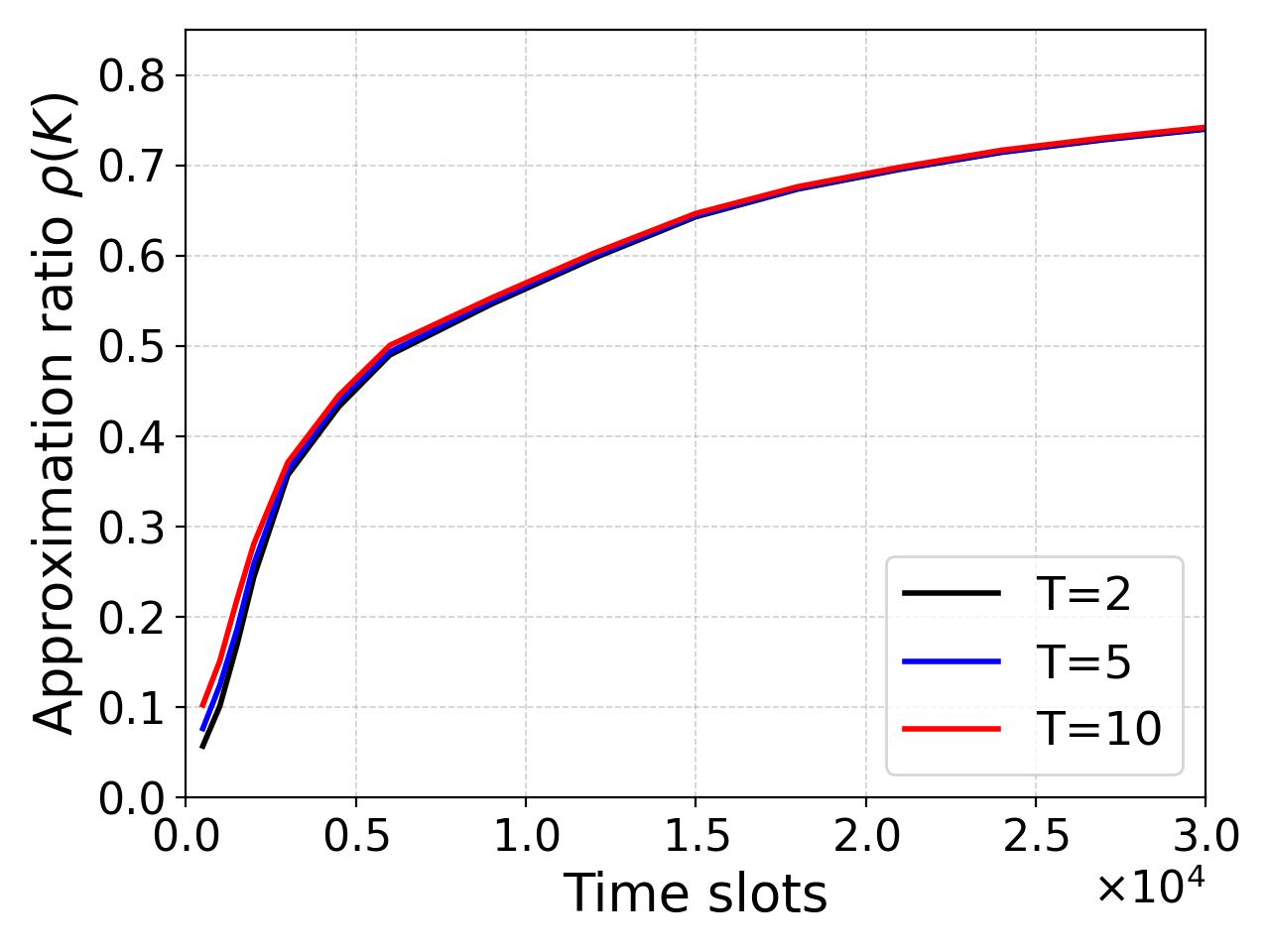}
}
\caption{Empirical approximation ratio $\rho(K)=1-\mathrm{Reg}(KT)/\mathrm{OPT}(KT)$ in trace-based simulation.}
\label{fig:approx}
\end{figure*}

\subsubsection{\hc{}}
Fig.~\ref{fig:trace} shows the performance of the \hc{} algorithm using our collected wireless channel traces. We observe from Fig.~\ref{fig:trace}\subref{fig:trace:fraction} that each user's scheduling fraction is larger than its desired value under different frame sizes, which shows that our \hc{} algorithm can guarantee long-term fairness. In addition, as shown in Fig.~\ref{fig:trace}\subref{fig:trace:fairness}, our algorithm can also achieve zero cumulative fairness violation and thus yields short-term fairness. A larger frame size results in a long time to achieve zero cumulative fairness and a larger amount of cumulative fairness violation. Moreover, the \hc{} algorithm can achieve sublinear regret and the frame size has a negligible impact on the regret performance. The impact of the slackness constant $\delta$ has a marginal impact on cumulative fairness violation, as shown in Fig.~\ref{fig:deltatrace}.
We further report the empirical approximation ratio $\rho(K)=\mathrm{ALG}(KT)/\mathrm{OPT}(KT)=1-\mathrm{Reg}(KT)/\mathrm{OPT}(KT)$. As shown in Fig.~\ref{fig:approx}\subref{fig:approx:outta}, the approximation ratio of the \hc{} algorithm increases monotonically with the horizon, attaining $\rho\approx0.74$ within the simulated horizon and continuing to grow toward $1$ as $K$ increases.

\subsubsection{\lc{}}
Fig.~\ref{fig:pac_trace} and Fig.~\ref{fig:pac_trace:delta} show the influence of frame size and slackness constant on the performance of the \lc{} algorithm, respectively, using our collected wireless channel traces. From Fig.~\ref{fig:pac_trace}\subref{fig:pac_trace:fraction}, we observe that each user's scheduling fraction exceeds its desired value across different frame sizes, demonstrating that our \lc{} algorithm ensures long-term fairness.
Additionally, Fig.~\ref{fig:pac_trace}\subref{fig:pac_trace:fairness} indicates a larger frame size results in a longer time required to achieve zero cumulative fairness violation and a larger amount of cumulative fairness violations. Compared with Fig.~\ref{fig:trace}\subref{fig:trace:fairness}, the low-complexity \lc{} algorithm requires a longer time to achieve zero cumulative fairness violation than the MaxWeight-type \hc{} algorithm. Similarly, as shown in Fig.~\ref{fig:pac_trace}\subref{fig:pac_trace:regret}, a larger frame size increases cumulative regret. Compared with Fig.~\ref{fig:trace}\subref{fig:trace:regret}, the low-complexity \lc{} algorithm generates larger cumulative regret than the MaxWeight-type \hc{} algorithm. It is also evident that the regret performance of the \lc{} algorithm is more sensitive to frame size than that of the \hc{} algorithm. The influence of the slackness constant $\delta$ on cumulative fairness violation and regret performance is marginal in Fig.~\ref{fig:pac_trace:delta}.
Fig.~\ref{fig:approx}\subref{fig:approx:pac} shows that the low-complexity \lc{} algorithm attains essentially the same empirical approximation ratio ($\rho\approx0.74$ within the simulated horizon) as the MaxWeight-type \hc{} algorithm, again increasing toward $1$ with the horizon and insensitive to the frame size $T$. Together with the cumulative-fairness-violation results above, this quantifies the complexity-performance trade-off of the pick-and-compare design: \lc{} matches \hc{} in throughput approximation ratio at a substantially lower computational cost, at the expense of achieving zero violation point slower.

\section{Conclusion}
In this paper, we studied the joint design of user association and scheduling and rate adaptation with different time scales in wireless networks to maximize cumulative throughput while guaranteeing desired fairness among users. We developed a MaxWeight-type user association and scheduling algorithm that combines both the virtual queues and UCB estimates in its weight measure. Each selected user utilizes the UCB algorithm to determine a transmission rate on a small time scale. We showed that our proposed algorithm yields $O(\sqrt{K\log K})$ cumulative regret over $K$ time frames and achieves zero cumulative fairness violation after a certain number of time frames. Considering the complexity of the MaxWeight-type algorithm, we also give a low-complexity algorithm based on the pick-and-compare design. We performed simulations to demonstrate the efficiency of these two proposed algorithms based on real-world data.

\appendices

\section{Proof of Proposition \ref{prop:violation}}
\label{APP:prop:violation}

Select the Lyapunov function 
\begin{align*}
    V(kT)\triangleq\|\mb{Q}(kT)\|,
\end{align*}
and consider its conditional expected drift given the current state $\mb{W}(kT)\triangleq(\mb{Q}(kT),\bs{w}(kT))$. We have the following key lemma.
\begin{lemma}
\label{lemma:drift}
For any $0<\epsilon_k\leq\delta/2$, if $V(kT)\geq W_{k}\triangleq (6N+N\delta^2+4\eta_kNLTr_M)/\delta$, then 
\begin{align}
    \bE\left[V((k+1)T)-V(kT)\middle|\mb{W}(kT)\right]\leq-\frac{\delta}{4}.
\end{align}
Moreover, the absolute drift of $V(kT)$ is bounded by $3N$, i.e., 
\begin{align}
  \left|V((k+1)T)-V(kT)\right|\leq 3N.  
\end{align}
\end{lemma}

The proof follows a similar line of arguments in \cite[Lemma 2]{eryilmaz2012asymptotically} and its 
proof is available in Appendix \ref{App:pf:lemma:drift}.

For any $\epsilon_k\leq\delta/2$, Lemma \ref{lemma:drift} satisfies the conditions of \cite[Lemma 11]{liu2021efficient} and thus we have 
\begin{align}
\bE\left[e^{\theta V(kT)}\right]\leq e^{\theta V(k'T)}+\frac{8e^{\theta(3N+W_{k})}}{\theta\delta},
\end{align}
where $\theta=\delta/(36N^{2}+N\delta)$ and $k'$ is the index such that for any $k\geq k'$, $\epsilon_k\leq\delta/2$. Note that such $k'$ exists due to the fact that $\{\epsilon_k\}_{k\geq0}$ is an decreasing sequence and $W_k$ is increasing when $k\geq k'$.

According to Jensen's inequality, for convex function $e^{\theta x}$, we have 
\begin{align}
e^{\theta\bE[V(kT)]}\leq\bE\left[e^{\theta V(kT)}\right]\leq e^{\theta\|\mb{Q}(k'T)\|}+\frac{8e^{\theta(3N+W_{k})}}{\theta\delta},
\end{align}
which implies that 
\begin{align}
\bE[V(kT)]\leq\frac{1}{\theta}\log\left(e^{\theta\|\mb{Q}(k'T)\|}+\frac{8e^{\theta(3N+W_{k})}}{\theta\delta}\right).    
\end{align}
Noting that  $V(kT)=\|\mb{Q}(kT)\|$ and  $\|\mb{Q}(kT)\|\geq\frac{1}{\sqrt{N}}\|\mb{Q}(kT)\|_1$, we have 
\begin{align}
\label{eqn:prop:vol:virtualQ}
&\bE[\|\mb{Q}(kT)\|_1]\leq\frac{\sqrt{N}}{\theta}\log\left(e^{\theta\|\mb{Q}(k'T)\|}+\frac{8e^{\theta(3N+W_{k})}}{\theta\delta}\right)\nonumber\\
\stackrel{(a)}{\leq}&\frac{\sqrt{N}}{\theta}\log\left(e^{\theta\|\mb{Q}(k'T)\|}+\frac{296N^2e^{\theta(3N+W_{k})}}{\delta^2}\right)\nonumber\\
\stackrel{(b)}{\leq}&\frac{\sqrt{N}}{\theta}\log\left(\frac{297N^2}{\delta^2}e^{\theta(3N+W_{k}+\|\mb{Q}(k'T)\|)}\right)\nonumber\\
=&\frac{\sqrt{N}}{\theta}\log\left(\frac{297N^2}{\delta^2}\right)+\sqrt{N}(3N+W_{k}+\|\mb{Q}(k'T)\|)\nonumber\\
\stackrel{(c)}{\leq}&\frac{74N^{2.5}}{\delta}\log\left(\frac{18N}{\delta}\right)+3N^{1.5}\nonumber\\
&\qquad+\frac{N^{1.5}(6+\delta^2+4\eta_kLTr_M)}{\delta}\nonumber\\
&+\sqrt{N}\|\mb{Q}(k'T)\|,
\end{align}
where step $(a)$ uses the fact that $1/\theta=(36N^2+N\delta)/\delta\leq 37N^2/\delta$ since $\delta<N$; $(b)$ follows from the fact that $\delta< N$; $(c)$ uses the fact that $297<324=18^2$, $1/\theta\leq 37N^2/\delta$, and the definition of $W_k$. 

According to the dynamics of virtual queues (cf. \eqref{eqn:virtualQ}), we have
\begin{align}
Q_n((k+1)T)\leq Q_n(kT)+\lambda_n+\epsilon_k, \forall \tau.
\end{align}
By summing the above inequality over $k=0,1,\ldots,k'-1$ and utilizing the fact that $Q_n(0)=0$, we have 
\begin{align}
\label{eqn:prop:vol:virtualQ:temp:lb}
Q_n(k'T)\leq k'\lambda_n+\sum_{k=0}^{k'-1}\epsilon_{k}\leq k'+\sum_{k=0}^{k'-1}\epsilon_{k},    
\end{align}
where the last step follows from the fact that $\lambda_n\leq1$.

Utilizing the fact that $\|\mb{Q}(k'T)\|\leq\|\mb{Q}(k'T)\|_1=\sum_{n=1}^{N}Q_n(k'T)$ and \eqref{eqn:prop:vol:virtualQ:temp:lb}, \eqref{eqn:prop:vol:virtualQ} becomes 
\begin{align}
\label{eqn:prop:vol:virtualQ_new}
&\bE[\|\mb{Q}(kT)\|_1]\leq\frac{74N^{2.5}}{\delta}\log\left(\frac{18N}{\delta}\right)+3N^{1.5}\nonumber\\
&+\frac{N^{1.5}(6+\delta^2+4\eta_kLTr_M)}{\delta}+N^{1.5}k'+N^{1.5}\sum_{k=0}^{k'-1}\epsilon_{k}.  
\end{align}

According to the dynamics of virtual queues (cf. \eqref{eqn:virtualQ}), we have
\begin{align}
Q_n((k+1)T)\geq Q_n(kT)+\lambda_n-\sum_{l}S_{l,n}(kT)+\epsilon_{k}, \forall k.
\end{align}
By summing the above inequality over $k=0,1,\ldots,K-1$, we have 
\begin{align}
\label{eqn:prop:vol:virtualQ:lb}
Q_n(KT)\geq\sum_{k=0}^{K-1}\left(\lambda_n-\sum_{l}S_{l,n}(kT)\right)+\sum_{k=0}^{K-1}\epsilon_k, \forall n.   
\end{align}
Noting that the user schedule is fixed within a frame and the virtual queue-lengths are only updated at the beginning of each frame, we have 
\begin{align}
\label{eqn:prop:vol:virtualQ:lb1}
&\left(T\bE[Q_n(KT)]-T\sum_{k=0}^{K-1}\epsilon_{k}\right)^{+}\nonumber\\
&\geq\left(\bE\left[\sum_{k=0}^{K-1}\sum_{t=kT}^{(k+1)T-1}\left(\lambda_n-\sum_{l}S_{l,n}(t)\right)\right]\right)^{+}, \forall n, 
\end{align}
where we utilize the fact that the virtual queue-lengths are non-negative. Hence, we have 
\begin{align}
\label{eqn:prop:vol:cumvol}
&\left(\bE\left[\sum_{k=0}^{K-1}\sum_{t=kT}^{(k+1)T-1}\left(\lambda_n-\sum_{l}S_{l,n}(t)\right)\right]\right)^{+}\nonumber\\
\leq&\left(T\bE\left[\|\mb{Q}(KT)\|_1\right]-T\sum_{k=0}^{K-1}\epsilon_{k}\right)^{+}\nonumber\\
\leq&T\bigg(\frac{74N^{2.5}}{\delta}\log\left(\frac{18N}{\delta}\right)+3N^{1.5}+N^{1.5}k'+N^{1.5}\sum_{l=0}^{k'-1}\epsilon_{l}\nonumber\\
&+\frac{N^{1.5}(6+\delta^2+4\eta_kLTr_M)}{\delta}-\sum_{k=0}^{K-1}\epsilon_{k}\bigg)^{+},
\end{align}
where the last step utilizes \eqref{eqn:prop:vol:virtualQ_new}. 

We will provide lower bound on $\sum_{k=0}^{K-1}\epsilon_{k}$ and upper bound on $\sum_{k=0}^{k'-1}\epsilon_{k}$, respectively. 
\begin{align}
\label{eqn:prop:vol:epsilon1}
\sum_{k=0}^{K-1}\epsilon_{k}=&\frac{4r_MLN^{1.5}+1}{2}\sum_{k=0}^{K-1}\frac{1}{\sqrt{k+1}}\nonumber\\
=&\frac{4r_MLN^{1.5}+1}{2}\sum_{k=1}^{K}\frac{1}{\sqrt{k}} \nonumber\\
\geq&\frac{4r_MLN^{1.5}+1}{2}\int_{1}^{K}\frac{1}{\sqrt{x}}dx\nonumber\\
=&(4r_MLN^{1.5}+1)(\sqrt{K}-1).
\end{align}
and 
\begin{align}
\label{eqn:prop:vol:epsilon2}
\sum_{k=0}^{k'-1}\epsilon_{k}=&\frac{4r_MN^{1.5}+1}{2}
\sum_{k=0}^{k'-1}\frac{1}{\sqrt{k+1}}\nonumber\\
=&\frac{4r_MLN^{1.5}+1}{2}\sum_{k=1}^{k'}\frac{1}{\sqrt{k}}\nonumber\\
\leq&\frac{4r_MLN^{1.5}+1}{2}\left(1+\int_{2}^{k'+1}\frac{1}{\sqrt{x-1}}dx\right)\nonumber\\
\leq&\frac{4r_MLN^{1.5}+1}{2}\left(2\sqrt{k'}-1\right)\nonumber\\
\leq&(4r_MLN^{1.5}+1)\sqrt{k'}.
\end{align}

We also note that $k'$ is the minimum integer such that $\epsilon_k\leq\delta/2$ and thus we have
\begin{align}
\label{eqn:prop:vol:epsilon3}
k'\leq(4r_MLN^{1.5}+1)^2/\delta^2.    
\end{align}
By substituting \eqref{eqn:prop:vol:epsilon1}, \eqref{eqn:prop:vol:epsilon2}, and \eqref{eqn:prop:vol:epsilon3} into \eqref{eqn:prop:vol:cumvol}, we have 
\begin{align*}
\left(\sum_{k=0}^{K-1}\sum_{t=kT}^{(k+1)T-1}\left(\lambda_n-S_n(t)\right)\right)^{+}\leq T(g(N,\delta,r_M)-\sqrt{k})^{+},
\end{align*}
where $g(N,\delta,r_M)=\frac{74N^{2.5}}{\delta}\log\left(\frac{18N}{\delta}\right)+(4r_ML+3)N^{1.5}+\frac{N^{1.5}(6+\delta^2)}{\delta}+1+\frac{N^{1.5}(4r_MLN^{1.5}+1)^2}{\delta}\left(\frac{1}{\delta}+1\right)$.

\section{Proof of Proposition \ref{prop:regret}}
\label{APP:prop:regret}
We rewrite the regret of the \hc{} algorithm as follows.
\begin{align}
&\text{Reg}(KT) 
\triangleq\sum_{k,l,n}\bE\left[S_{l,n}^*\sum_{t=k T}^{(k+1)T-1}\sum_{m=1}^{M}r_m\mu_{l,n,m}I_{l,n,m}^*\right] \nonumber\\
- & \sum_{k,l,n}\bE\left[\wh{S}_{l,n}(kT)\sum_{t=k T}^{(k+1)T-1}\sum_{m=1}^{M}r_m\mu_{l,n,m}\wh{I}_{l,n,m}(t)\right]\nonumber\\
=&\sum_{k=0}^{K-1}\Delta R(kT),
\end{align}
where
\begin{align*}
\Delta R(kT)\triangleq \sum_{t=kT}^{(k+1)T-1}\sum_{l,n,m}
\bE\big[&r_{m}\mu_{l,n,m} S_{l,n}^*I_{l,n,m}^*\\
&- r_{m}\mu_{l,n,m}\wh{S}_{l,n}(kT)\wh{I}_{l,n,m}(t)\big].
\end{align*}

Select the Lyapunov function $V_1(\mb{Q})\triangleq\frac{1}{2}\sum_{n=1}^{N}Q_n^2$ and consider its expected drift. In the rest of the proof, we omit the frame index $kT$ associated with virtual queue lengths $\mb{Q}$ and schedule $\mathbf{S}$ without causing ambiguity. 
\begin{align}
\label{eqn:prop:reg:lyapunov}
&\bE[V_1(\mb{Q}((k+1)T)-V_1(\mb{Q}(kT))]\nonumber\\
\stackrel{(a)}{\leq}&\frac{1}{2}\sum_{n=1}^{N}\bE\left[\left(Q_n + \lambda_n -\sum_{l}\wh{S}_{l,n} +\epsilon_k\right)^2 - \frac{1}{2}\sum_{n=1}^{N}Q_n^2\right]\nonumber\\
=&\sum_{n=1}^{N}\bE\left[Q_n\left(\lambda_n + \epsilon_k - \sum_{l}\wh{S}_{l,n}\right)\right]\nonumber\\
&+ \frac{1}{2}\sum_{n=1}^{N}\bE\left[\left(\lambda_n - \sum_{l}\wh{S}_{l,n}+\epsilon_k\right)^2\right]\nonumber\\
\stackrel{(b)}{\leq}&\sum_{n=1}^{N}\bE\left[(\lambda_n+\epsilon_k)Q_n\right] - \sum_{l=1}^{L}\sum_{n=1}^{N}\bE\left[Q_n\wh{S}_{l,n}\right] + H_k,
\end{align}
where step $(a)$ uses the fact that $\left(\max\{x,0\}\right)^2\leq x^2$; $(b)$ is true for $H_k\triangleq N\left(1+\epsilon_k^2/2\right)$.

Adding the term $\eta_k\Delta R(kT)$ on both sides of \eqref{eqn:prop:reg:lyapunov} and utilizing the fact that the optimal stationary randomized policy $\mb{S}^*(kT)$ is independent of the system state and stabilizes the system, i.e., $\bE[\sum_{l=1}^{L}S_{l,n}^*(kT)]\geq\lambda_n+\delta, \forall n$, we have 
\begin{align}
\label{eqn:prop:reg:lyapunov+}
&\bE[V_1(\mb{Q}((k+1)T)-V_1(\mb{Q}(kT))] + \eta_k\Delta R(kT)
\nonumber\\
&\leq\sum_{n}\bE[(\lambda_n+\epsilon_k)Q_n] - \sum_{l,n}\bE[Q_n\wh{S}_{l,n}] + H_k \nonumber\\
&+ \eta_k\sum_{t=kT}^{(k+1)T-1}\sum_{l,n,m}\bE\big[r_m\mu_{l,n,m}S_{l,n}^*I_{l,n,m}^* \nonumber\\
&\qquad\qquad\qquad\qquad -r_m\mu_{l,n,m}\wh{S}_{l,n}\wh{I}_{l,n,m}(t)\big]\nonumber\\
&=H_k+\sum_{n}\bE\left[Q_n\left(\lambda_n +\epsilon_k- \sum_{l}S_{l,n}^*\right)\right]\nonumber\\
&+\sum_{l,n}\bE\Bigg[\left(Q_n+\eta_k\sum_{t=kT}^{(k+1)T-1}\sum_{m}r_m\mu_{l,n,m}I^*_{l,n,m}\right)\nonumber\\
&\qquad\qquad\qquad\qquad \cdot\left(S_{l,n}^*-\wh{S}_{l,n}\right)\Bigg]\nonumber\\
&+\eta_k\sum_{l,n}\sum_{t=kT}^{(k+1)T-1}\sum_{m}\bE\Big[r_m\mu_{l,n,m}\nonumber\\
&\qquad\qquad\qquad\cdot\left(I_{l,n,m}^*-\wh{I}_{l,n,m}(t)\right)\wh{S}_{l,n}\Big]\nonumber\\
&\leq H_k+\sum_{l,n}\bE\Big[\left(Q_n+\eta_k T\max_{m}r_m\mu_{l,n,m}\right)\nonumber\\
&\qquad\qquad\qquad\cdot\left(S_{l,n}^*-\wh{S}_{l,n}\right)\Big]\nonumber\\
&+\eta_k\sum_{l,n,m}\sum_{t=kT}^{(k+1)T-1}\bE\Big[r_m\mu_{l,n,m}\nonumber\\
&\qquad\qquad\qquad\cdot\left(I_{l,n,m}^*-\wh{I}_{l,n,m}(t)\right)\wh{S}_{l,n}\Big],
\end{align}
where the last step follows from the fact the optimal stationary randomized policy $\mb{S}^*(kT)$ is independent of the system state and stabilizes the system, i.e., $\bE[S_n^*(kT)]\geq\lambda_n+\delta, \forall n$ and 
holds for any $k\geq k_0\triangleq(4r_MLN^{1.5}+1)^2/4\delta^2$ such that $\epsilon_{k_0}\leq\delta$.

Dividing $\eta_k$ on both sides of \eqref{eqn:prop:reg:lyapunov+}, we have 
\begin{align}
\label{eqn:prop:reg:lyapunov+:v1}
&\frac{1}{\eta_k}\bE[V_1(\mb{Q}((k+1)T)-V_1(\mb{Q}(kT))]+\Delta R(kT)\leq\frac{H_k}{\eta_k}\nonumber\\
&+\frac{1}{\eta_k}\sum_{l,n}\bE\bigg[\left(Q_n+\eta_k T\max_{m}r_m\mu_{l,n,m}\right)\left(S_{l,n}^*-\wh{S}_{l,n}\right)\bigg]\nonumber\\
&+\sum_{l,n,m}\sum_{t=kT}^{(k+1)T-1}\bE\left[r_m\mu_{l,n,m}\left(I_{l,n,m}^*-\wh{I}_{l,n,m}(t)\right)\wh{S}_{l,n}\right] .
\end{align}

Summing \eqref{eqn:prop:reg:lyapunov+:v1} over $k=k_0,k_0+1,\ldots,K-1$ and we have 
\begin{align}
\label{eqn:prop:reg:ub}
&\text{Reg}(KT)=\sum_{k=0}^{k_0-1}\Delta R(kT)+\sum_{k=k_0}^{K-1}\Delta R(kT) \nonumber\\
&\stackrel{(a)}{\leq} k_0Nr_MT+\sum_{k=k_0}^{K-1}\frac{H_k}{\eta_k}+\frac{1}{\eta_{k_0}}V_1(k_0T)\nonumber\\
&+\sum_{k=k_0}^{K-1}\frac{1}{\eta_k}\sum_{l,n}\bE\bigg[\left(Q_n+\eta_k T\max_{m}r_m\mu_{l,n,m}\right)\nonumber\\
&\qquad\qquad\qquad\qquad\cdot\left(S_{l,n}^*-\wh{S}_{l,n}\right)\bigg]\nonumber\\
&+\sum_{k=k_0}^{K-1}\sum_{l,n,m}\sum_{t=kT}^{(k+1)T-1}\bE\bigg[r_m\mu_{l,n,m}\nonumber\\
&\qquad\qquad\qquad\cdot\left(I_{l,n,m}^*-\wh{I}_{l,n,m}(t)\right)\wh{S}_{l,n}\bigg] \nonumber\\
&\stackrel{(b)}{\leq}   k_0Nr_MT+\sum_{k=k_0}^{K-1}\frac{H_k}{\eta_k}\nonumber\\
&\qquad\qquad+ \frac{1}{2\eta_{k_0}}N\left(k_0+\sum_{k=0}^{k_0-1}\epsilon_{k}\right)^2 \nonumber\\
&+\sum_{k=k_0}^{K-1}\frac{1}{\eta_k}\sum_{l,n}\bE\bigg[\left(Q_n+\eta_k T\max_{m}r_m\mu_{l,n,m}\right)\nonumber\\
&\qquad\qquad\qquad\qquad\cdot\left(S_{l,n}^*-\wh{S}_{l,n}\right)\bigg]\nonumber\\
&+\sum_{k=k_0}^{K-1}\sum_{l,n,m}\sum_{t=kT}^{(k+1)T-1}\bE\bigg[r_m\mu_{l,n,m}\nonumber\\
&\qquad\qquad\qquad\cdot\left(I_{l,n,m}^*-\wh{I}_{l,n,m}(t)\right)\wh{S}_{l,n}\bigg], \nonumber\\
\end{align}
where step $(a)$ uses the fact that $\eta_{k}$ is an increasing sequence; $(b)$ follows from (\ref{eqn:prop:vol:virtualQ:temp:lb}).

Next, we focus on the term 
$$\sum_{l,n}\left(Q_n+\eta_k T\max_{m}r_m\mu_{l,n,m}\right)\left(S_{l,n}^*-\wh{S}_{l,n}\right).$$ Utilizing the scheduling component of the \hc{} algorithm, we have
\begin{align}
\label{eqn:prop:reg:ub:Q}
&\sum_{l,n}\left(Q_n+\eta_k T\max_{m}r_m\mu_{l,n,m}\right)\left(S_{l,n}^*-\wh{S}_{l,n}\right)\nonumber\\
\stackrel{(a)}{\leq}&\sum_{l,n}\left(Q_n+\eta_k T\max_{m}r_m\mu_{l,n,m}\right)\wt{S}_{l,n}\nonumber\\
&-\sum_{l,n}\left(Q_n+\eta_k T\max_{m}r_m\mu_{l,n,m}\right)\wh{S}_{l,n}\nonumber\\
\stackrel{(b)}{\leq}&\sum_{l,n}\left(Q_n+\eta_k T\max_{m}r_m\mu_{l,n,m}\right)\wt{S}_{l,n}\nonumber\\
&-\sum_{l,n}\left(Q_n+\eta_k T\max_{m}r_m\mu_{l,n,m}\right)\wh{S}_{l,n}\nonumber\\
&+\sum_{l,n}\left(Q_n+\eta_k T \max_{m}r_mw_{l,n,m}(kT)\right)\wh{S}_{l,n}\nonumber\\
&-\sum_{l,n}\left(Q_n+\eta_k T \max_{m}r_mw_{l,n,m}(kT)\right)\wt{S}_{l,n}\nonumber\\
=&\eta_k T\bigg(\sum_{l,n}(\max_{m}r_mw_{l,n,m}(kT)-\max_{m}r_m\mu_{l,n,m})\wh{S}_{l,n}\nonumber\\
&+\sum_{l,n}(\max_{m}r_m\mu_{l,n,m}-\max_{m}r_mw_{l,n,m}(kT))\wt{S}_{l,n}\bigg)\nonumber\\
\stackrel{(c)}{\leq}&\eta_kT\bigg(\sum_{l,n,m}(r_mw_{l,n,m}(kT)-r_m\mu_{l,n,m})\wh{I}_{l,n,m}(kT)\wh{S}_{l,n}\nonumber\\
&+\sum_{l,n}(\max_{m}r_m\mu_{l,n,m}-\max_{m}r_mw_{l,n,m}(kT))\wt{S}_{l,n}\bigg)\nonumber\\
\stackrel{(d)}{\leq}&\eta_k T\bigg(\sum_{l,n,m}r_m(w_{l,n,m}(kT)-\mu_{l,n,m})^{+}\wh{I}_{l,n,m}(kT)\nonumber\\
&+\sum_{l,n}\max_{m}r_m\left(\mu_{l,n,m}-w_{l,n,m}(kT)\right)\wt{S}_{l,n}\bigg)\nonumber\\
\stackrel{(e)}{\leq}&\eta_k T\sum_{l,n,m}r_m\left(w_{l,n,m}(kT)-\mu_{l,n,m}\right)^{+}\wh{I}_{l,n,m}(kT)\nonumber\\
&+\eta_k T\sum_{l,n,m}r_m\left(\mu_{l,n,m}-w_{l,n,m}(kT)\right)^{+}\wt{S}_{l,n},
\end{align}
where step $(a)$ is true for $\wt{\mb{S}}\triangleq(\wt{S}_{l,n})_{l,n}\in\argmax_{\mb{S}\in\mc{S}}\sum_{l,n}\left(Q_n+\eta_k T\max_{m}r_m\mu_{l,n,m}\right)S_{l,n}$; $(b)$ follows from the definition of the \hc{} algorithm; $(c)$ uses the fact that $\max_{m}r_mw_{l,n,m}(kT)=\sum_{m}r_mw_{l,n,m}(kT)\wh{I}_{l,n,m}(kT)$ and $\max_{m}r_m\mu_{l,n,m}(kT)=\sum_{m}r_m\mu_{l,n,m}(kT)I^*_{l,n,m}\geq\sum_{m}r_m\mu_{l,n,m}(kT)\wh{I}_{l,n,m}(kT)$; $(d)$ follows from the fact that $\max_{m=1,\ldots,M}x_m - \max_{m=1,\ldots,M}y_m\leq\max_{m=1,\ldots,M}(x_m-y_m)$; $(d)$ uses the fact that $\max_{m=1,\ldots,M}x_m\leq \max_{m=1,\ldots,M}(x_m)^{+}\leq\sum_{m=1}^{M}(x_m)^{+}$ and $(x)^{+}=\max\{x,0\}$.

For the term
\begin{align*}
\sum_{t=kT}^{(k+1)T-1}\sum_{m}r_m\mu_{l,n,m}\left(I_{l,n,m}^{*}-\wh{I}_{l,n,m}(t)\right),
\end{align*}
utilizing the rate adaptation component of the \hc{} algorithm, we have
\begin{align}
\label{eqn:prop:reg:ub:I}
&\sum_{t=kT}^{(k+1)T-1}\sum_{m}r_m\mu_{l,n,m}\left(I_{l,n,m}^*-\wh{I}_{l,n,m}(t)\right)\nonumber\\
\leq&\sum_{t=kT}^{(k+1)T-1}\sum_{m}r_m\mu_{l,n,m}\left(I_{l,n,m}^*-\wh{I}_{l,n,m}(t)\right)\nonumber\\
&+\sum_{t=kT}^{(k+1)T-1}\sum_{m}r_mw_{l,n,m}(t)\left(\wh{I}_{l,n,m}(t)-I_{l,n,m}^*\right)\nonumber\\
=&\sum_{t=kT}^{(k+1)T-1}\sum_{m}r_m(w_{l,n,m}(t)-\mu_{l,n,m})\wh{I}_{l,n,m}(t)\nonumber\\
&+\sum_{t=kT}^{(k+1)T-1}\sum_{m}r_m(\mu_{l,n,m}-w_{l,n,m}(t))I_{l,n,m}^{*}.
\end{align}

Substituting \eqref{eqn:prop:reg:ub:Q} and \eqref{eqn:prop:reg:ub:I} into \eqref{eqn:prop:reg:ub}, we have 
\begin{align}
\label{eqn:prop:reg:main:final}
&\text{Reg}(KT)\leq k_0Nr_MT+\sum_{k=k_0}^{K-1}\frac{H_k}{\eta_k}\nonumber\\
&\qquad\qquad+ \frac{1}{2\eta_{k_0}}N\left(k_0+\sum_{k=0}^{k_0-1}\epsilon_{k}\right)^2 \nonumber\\
&+\underbrace{\sum_{k=0}^{K-1}\sum_{l,n,m}\sum_{t=kT}^{(k+1)T-1}r_m\bE\left[\left(w_{l,n,m}(t)-\mu_{l,n,m}\right)\wh{I}_{l,n,m}(t)\right]}_{\triangleq G_1(KT)}\nonumber\\
&+ \underbrace{\sum_{k=0}^{K-1}\sum_{l,n,m}\sum_{t=kT}^{(k+1)T-1}r_m\bE\left[\left(\mu_{l,n,m}-w_{l,n,m}(t)\right)I_{l,n,m}^{*}\right]}_{\triangleq G_2(KT)}\nonumber\\
&+\underbrace{T\sum_{k=0}^{K-1}\sum_{l,n,m}r_m\bE\left[\left(w_{l,n,m}(kT)-\mu_{l,n,m}\right)^{+}\wh{I}_{l,n,m}(kT)\right]}_{\triangleq G_3(KT)} \nonumber\\
&+ \underbrace{T\sum_{k=0}^{K-1}\sum_{l,n,m}r_m\bE\left[\left(\mu_{l,n,m}-w_{l,n,m}(kT)\right)^{+}\wt{S}_{l,n}\right]}_{\triangleq G_4(KT)}.
\end{align}
For the term $\sum_{k=k_0}^{K-1}H_k/\eta_k$, according to the definition of ${\eta_k}$, we can show that
\begin{align}
\label{eqn:prop:reg:R0}
&\sum_{k=k_0}^{K-1}\frac{H_k}{\eta_k}\leq 2\sqrt{K}NT\left(\delta+\frac{3}{2\delta}\right).
\end{align}

Next, we focus on $G_1(KT)$, $G_2(KT)$, $G_3(KT)$ and $G_4(KT)$, respectively. Recall that $H_{l,n,m}(t)$ is the number of times that user $n$ is associated with AP $l$ and uses rate $r_m$ until time slot $t$. Let $t_{l,n,m,\tau}$ denote the time slot at which user $n$ is associated with AP $l$ and uses rate $r_m$, where $\tau=1,2,\ldots, H_{l,n,m}(KT)$. Therefore, we have $H_{l,n,m}(t_{l,n,m,\tau})=\tau-1$. Let
\begin{align*}
G_{l,n,m,1}(KT)\triangleq\sum_{k=0}^{K-1}\sum_{t=kT}^{(k+1)T-1}\bE\Big[&\left(w_{l,n,m}(t)-\mu_{l,n,m}\right)\\
&\cdot \wh{I}_{l,n,m}(t)\Big]
\end{align*}
and thus $G_1(KT)=\sum_{l,n,m}r_mG_{l,n,m,1}(KT)$. Then, we have 
\begin{align}
\label{eqn:prop:reg:R1}
&G_{l,n,m,1}(KT) \nonumber\\
\stackrel{(a)}{\leq}&\sum_{k=0}^{K-1}\sum_{t=kT}^{(k+1)T-1}\bE\left[(w_{l,n,m}(t)-\mu_{l,n,m})\wh{I}_{l,n,m}(t)\id_{\mc{F}_{l,n,m}(t)}\right] \nonumber\\
\stackrel{(b)}{\leq}&\bE\bigg[\sum_{\tau=1}^{H_{l,n,m}(KT)}(w_{l,n,m}(t_{l,n,m,\tau})-\mu_{l,n,m})\nonumber\\
&\qquad\qquad\qquad\qquad\cdot\id_{\mc{F}_{l,n,m}(t_{l,n,m,\tau})}\bigg] \nonumber\\
\stackrel{(c)}{\leq}&1+\bE\bigg[\sum_{\tau=2}^{H_{l,n,m}(KT)}(w_{l,n,m}(t_{l,n,m,\tau})-\mu_{l,n,m})\nonumber\\
&\qquad\qquad\qquad\qquad\cdot\id_{\mc{F}_{l,n,m}(t_{l,n,m,\tau})}\bigg]\nonumber\\
\stackrel{(d)}{\leq}&1+\sum_{\tau=2}^{\infty}\bE\left[\id_{\ol{\mc{G}}_{l,n,m}(t_{l,n,m,\tau)}}\right]\nonumber\\
+&\bE\Bigg[\sum_{\tau=2}^{H_{l,n,m}(KT)}(w_{l,n,m}(t_{l,n,m,\tau})-\mu_{l,n,m})\nonumber\\
&\qquad\qquad\qquad\qquad\cdot\id_{\mc{F}_{l,n,m}(t_{l,n,m,\tau})\cap\mc{G}_{l,n,m}(t_{l,n,m,\tau})}\Bigg],
\end{align}
where step $(a)$ is true for $\mc{F}_{l,n,m}(t)\triangleq\{w_{l,n,m}(t)\geq\mu_{l,n,m}\}$ and $\id_{\{\cdot\}}$ being an indicator function; $(b)$ uses the definition of $t_{l,n,m,\tau}$; $(c)$ follows from the fact that $w_{l,n,m}(t)\leq1, \forall t\geq0$; $(d)$ is true for 
$$\mc{G}_{l,n,m}(t)\triangleq\left\{\ol{\mu}_{l,n,m}(t)-\mu_{l,n,m}\leq\sqrt{\frac{3\log t}{2H_{l,n,m}(t)}}\right\},$$
and $\ol{\mc{G}}_{l,n,m}(t)$ being the complement of the event $\mc{G}_{l,n,m}(t)$.

Next, we consider the second term on the right hand side (RHS) of \eqref{eqn:prop:reg:R1}.

\begin{align*}
&\bE\left[\id_{\ol{\mc{G}}_{l,n,m}(t_{l,n,m,\tau})}\right]=\Pr\{\ol{\mc{G}}_{l,n,m}(t_{l,n,m,\tau})\}\nonumber\\
\stackrel{(a)}{\leq}&\Pr\left\{\bigcup_{\nu=\tau-1}^{KT-1}\left\{\ol{\mu}_{l,n,m}(\nu)-\mu_{l,n,m}>\sqrt{\frac{3\log \nu}{2(\tau-1)}}\right\}\right\}\nonumber\\
\stackrel{(b)}{\leq}&\sum_{\nu=\tau-1}^{KT-1}\Pr\left\{\ol{\mu}_{l,n,m}(\nu)-\mu_{l,n,m}>\sqrt{\frac{3\log\nu}{2(\tau-1)}}\right\} \nonumber\\
\stackrel{(c)}{\leq}&\sum_{\nu=\tau-1}^{KT-1}\frac{1}{\nu^3}\leq\frac{1}{(\tau-1)^3}+\int_{\tau-1}^{\infty}\frac{1}{(x-1)^3}dx\stackrel{(d)}{\leq}\frac{3}{2(\tau-1)^2},
\end{align*}
where step $(a)$ follows from the fact that 
\begin{align*}
\ol{\mc{G}}_{l,n,m}(t_{l,n,m,\tau})\subset\bigcup_{\nu=\tau-1}^{KT-1}\bigg\{&\ol{\mu}_{l,n,m}(\nu)-\mu_{l,n,m}\\
&>\sqrt{\frac{3\log\nu}{2(\tau-1)}}\bigg\};
\end{align*}
$(b)$ uses the union bound; $(c)$ follows from the Chernoff-Hoeffding Bound (see, e.g., \cite[Fact 1]{auer2002finite}), i.e., for $X_1,X_2,\ldots,X_n$ be i.i.d. random variables with common range $[0,1]$ and mean $\mu$, then for any $a\geq0$, we have 
\begin{align}
\label{eqn:prop:reg:chernoff}
\Pr\left\{\frac{1}{n}\sum_{i=1}^{n}X_i\geq \mu+a\right\}\leq e^{-2na^2},
\end{align}
$(d)$ is true for $\tau\geq2$.

Hence, the third term on the RHS of \eqref{eqn:prop:reg:R1} can be bounded as follows.
\begin{align}
\label{eqn:prop:reg:R1:third}
\bE\left[\sum_{\tau=2}^{\infty}\id_{\ol{\mc{G}}_{l,n,m}(t_{l,n,m,\tau})}\right]
\leq\sum_{\tau=1}^{\infty}\frac{3}{2\tau^2}=\frac{\pi^2}{4},
\end{align}
where the last step use the fact that $\sum_{n=1}^{\infty}1/n^2=\pi^2/6$.

With regard to the third term on the RHS of \eqref{eqn:prop:reg:R1}, we have 
\begin{align}
\label{eqn:prop:reg:R1:second}
&\bE\Bigg[\sum_{\tau=2}^{H_{l,n,m}(KT)}(w_{l,n,m}(t_{l,n,m,\tau})-\mu_{l,n,m}) \nonumber \\
& \quad\quad\quad\quad\quad\quad\quad \cdot\id_{\mc{F}_{l,n,m}(t_{l,n,m,\tau})\mc{G}_{l,n,m}(t_{l,n,m,\tau})}\Bigg] \nonumber\\
&\stackrel{(a)}{\leq}\bE\left[\sum_{\tau=2}^{H_{l,n,m}(KT)}2\sqrt{\frac{3\log t_{l,n,m,\tau}}{2H_{l,n,m}(t_{l,n,m,\tau})}}\right] \nonumber\\
&\stackrel{(b)}{\leq}\sqrt{6\log (KT)}\bE\left[\sum_{\tau=2}^{H_{l,n,m}(KT)}\frac{1}{\sqrt{\tau-1}}\right]\nonumber\\
&\leq\sqrt{6\log(KT)}\bE\left[1 + \int_{1}^{H_{l,n,m}(KT)}\frac{1}{\sqrt{x}}dx\right]\nonumber\\
&\leq2\sqrt{6\log(KT)}\bE\left[\sqrt{H_{l,n,m}(KT)}\right],
\end{align}
where step $(a)$ uses the definition of $w_{l,n,m}(t)$ and $\mc{G}_{l,n,m}(t)$, and $(b)$ follows from the fact that $t_{l,n,m,\tau}\leq KT$ and the definition of $t_{l,n,m,\tau}$. 

By substituting \eqref{eqn:prop:reg:R1:second} and \eqref{eqn:prop:reg:R1:third} into \eqref{eqn:prop:reg:R1} and using the definition of $G_1(KT)$, we have 
\begin{align}
\label{eqn:prop:reg:R1:final}
&G_1(KT)\leq LNMr_M\left(1+\frac{\pi^2}{4}\right) \nonumber\\
&+ 2r_M\sqrt{6\log (KT)}\sum_{l,n,m}\bE\left[\sqrt{H_{l,n,m}(KT)}\right]\nonumber\\
\stackrel{(a)}{\leq}&LNMr_M\left(1+\frac{\pi^2}{4}\right) \nonumber\\
&+ 2LNMr_M\sqrt{6\log (KT)}\nonumber\\
&\qquad\qquad\cdot\bE\left[\sqrt{\frac{1}{LNM}\sum_{l,n,m}H_{l,n,m}(KT)}\right] \nonumber\\
\stackrel{(b)}{\leq}&LNMr_M\left(1+\frac{\pi^2}{4}\right)\nonumber\\
&+ 2r_M\sqrt{6LNMS_{\max}KT\log (KT)},
\end{align}
where step $(a)$ uses the Jensen's inequality and the concavity of the function $\sqrt{x}$, and $(b)$ is true since $\sum_{l,n,m}H_{l,n,m}(KT)\leq KTS_{\max}$ and $S_{\max}$ is the maximum number of users that can be scheduled in each time slot.

Next, we consider the term $G_2(KT)$. First, we note that  
\begin{align}
& G_2(KT) \nonumber \\ 
\leq&\sum_{k=0}^{K-1}\sum_{l,n,m}\sum_{t=kT}^{(k+1)T-1}r_m\bE\bigg[(\mu_{l,n,m}-w_{l,n,m}(t))\nonumber\\
&\qquad\qquad\qquad\qquad\qquad\cdot I_{l,n,m}^{*}\id_{\ol{\mc{F}}_{l,n,m}(t)}\bigg],
\end{align}
where we recall that $\mc{F}_{l,n,m}(t)\triangleq\{w_{l,n,m}(t) \geq \mu_{l,n,m} \}$. Note that for $t\leq t_{l,n,m,1}$, we have $w_{l,n,m}(t)=1$ and thus $\mc{F}_{l,n,m}(t)$ happens. Therefore, we have 
\begin{align}
\label{eqn:prop:reg:R2:final}
& G_2(KT) \nonumber\\
\leq &r_M\sum_{l,n,m}\bE\bigg[\sum_{t=t_{l,n,m,1}+1}^{KT-1}\left(\mu_{l,n,m}-w_{l,n,m}(t)\right)\nonumber\\
&\qquad\qquad\qquad\qquad\cdot I_{l,n,m}^{*}\id_{\ol{\mc{F}}_{l,n,m}(t)}\bigg]\nonumber\\
\stackrel{(a)}{\leq}&r_M\sum_{l,n,m}\sum_{t=t_{l,n,m,1}+1}^{KT-1}\nonumber\\
&\qquad\quad\Pr\bigg\{\ol{\mu}_{l,n,m}(t)-\mu_{l,n,m}\leq-\sqrt{\frac{3\log t}{2H_{l,n,m}(t)}}\bigg\}\nonumber\\
\leq &r_M\sum_{l,n,m}\sum_{\tau=1}^{KT-1}\sum_{\nu=1}^{\tau}\nonumber\\
&\qquad\quad\Pr\bigg\{\frac{1}{\nu}\sum_{i=1}^{\nu}X(i)-\mu_{l,n,m}\leq-\sqrt{\frac{3\log \tau}{2l}}\bigg\}\nonumber\\
\stackrel{(b)}{\leq}&r_M\sum_{l,n,m}\sum_{\tau=1}^{KT-1}\sum_{\nu=1}^{\tau}\frac{1}{\tau^3}\nonumber\\
=&r_M\sum_{l,n,m}\sum_{\tau=1}^{KT-1}\frac{1}{\tau^2}\stackrel{(c)}{\leq}\frac{LNMr_M\pi^2}{6},
\end{align}
where step $(a)$ follows from the definition of $\ol{\mc{F}}_{l,n,m}(t)$ and the fact that $\mu_{l,n,m}\leq 1$ and $I^*_{l,n,m}\leq 1$; $(b)$ again uses the Chernoff-Hoeffding Bound (cf. \eqref{eqn:prop:reg:chernoff}); $(c)$ is true since $\sum_{\tau=1}^{KT-1}1/\tau^2\leq\sum_{\tau=1}^{\infty}1/\tau^2=\pi^2/6$. 

Next, we will derive the upper bound of $G_3(KT)$ and $G_4(KT)$. Let
\begin{align*}
G_{l,n,m,3}(KT)\triangleq \sum_{k=0}^{K-1}\bE\Big[&\left(w_{l,n,m}(kT)-\mu_{l,n,m}\right)^{+}\\
&\cdot\wh{I}_{l,n,m}(kT)\Big]
\end{align*}
and thus $G_3(KT)=T\sum_{l,n,m}r_mG_{l,n,m,3}(KT)$. 
Next, we analyze the term $G_{l,n,m,3}(KT)$. 
\begin{align}
\label{eqn:prop:reg:R3}
& G_{l,n,m,3}(KT) \nonumber \\
& =\sum_{k=0}^{K-1}\bE\left[(w_{l,n,m}(kT)-\mu_{l,n,m})\wh{I}_{l,n,m}(kT)\id_{\mc{F}_{l,n,m}(kT)}\right] \nonumber\\
&\leq\sum_{k=0}^{K-1}\sum_{t=kT}^{(k+1)T-1}\bE\left[\left(w_{l,n,m}(t)-\mu_{l,n,m}\right)\wh{I}_{l,n,m}(t)\id_{\mc{F}_{l,n,m}(t)}\right]\nonumber\\
&\leq 1+\frac{\pi^2}{4}+2\sqrt{6\log(KT)}\bE\left[\sqrt{H_{l,n,m}(KT)}\right],
\end{align}
where the last inequality uses the derived upper bound on $G_{l,n,m,1}(kT)$. 

Hence, we have 
\begin{align}
&G_3(KT)
\leq LNMr_M\left(1+\frac{\pi^2}{4}\right)T\nonumber\\
& + 2r_M T\sqrt{6\log (KT)}\sum_{l,n,m}\bE\left[\sqrt{H_{l,n,m}(KT)}\right]\nonumber\\
&\stackrel{(a)}{\leq}LNMr_M\left(1+\frac{\pi^2}{4}\right)T +2LNMr_M T\sqrt{6\log (KT)} \nonumber \\
&\qquad\qquad\qquad\qquad\qquad\cdot\bE\left[\sqrt{\frac{1}{LNM}\sum_{l,n,m}H_{l,n,m}(KT)}\right]\nonumber\\
&\stackrel{(b)}{\leq} LNMr_M\left(1+\frac{\pi^2}{4}\right)T \nonumber\\
& + 2r_M T\sqrt{6LNMS_{\max}KT\log (KT)}, \label{eqn:prop:reg:R3:final}
\end{align}
where step $(a)$ uses the Jensen's inequality, and $(b)$ is true since $\sum_{l,m,n}H_{l,n,m}(KT)\leq KTS_{\max}$ and $S_{\max}$ is the maximum number of users that can be scheduled in each time slot.

Next, we consider the term $G_4(KT)$.
\begin{align}
\label{eqn:prop:reg:R4:final}
&G_4(KT)\stackrel{(a)}{\leq} r_M T\sum_{l,n,m}\sum_{k=0}^{K-1}\bE\bigg[(\mu_{l,n,m}-w_{l,n,m}(kT)) \nonumber \\
& \qquad\qquad\qquad\qquad\qquad\qquad\qquad \cdot\id_{\ol{\mc{F}}_{l,n,m}(kT)}\bigg]\nonumber\\
\leq &r_M T\sum_{l,n,m}\sum_{\tau=1}^{KT-1}\sum_{\nu=1}^{\tau}\nonumber\\
&\qquad\quad\Pr\bigg\{\frac{1}{\nu}\sum_{i=1}^{\nu}X(i)-\mu_{l,n,m}\leq-\sqrt{\frac{3\log \tau}{2l}}\bigg\}\nonumber\\
\leq&r_M T\sum_{l,n,m}\sum_{\tau=1}^{KT-1}\sum_{\nu=1}^{\tau}\frac{1}{\tau^3}\nonumber\\
=&r_M T\sum_{l,n,m}\sum_{\tau=1}^{KT-1}\frac{1}{\tau^2}\leq\frac{LNMr_M T\pi^2}{6},
\end{align}
where step $(a)$ uses the fact $\wt{S}_{l,n} \leq 1,\forall l,\forall n$, and other inequalities in \eqref{eqn:prop:reg:R4:final} are similar to those in \eqref{eqn:prop:reg:R2:final}. 

Hence, by substituting \eqref{eqn:prop:reg:R0}, \eqref{eqn:prop:reg:R1:final}, \eqref{eqn:prop:reg:R2:final},  \eqref{eqn:prop:reg:R3:final}, and  \eqref{eqn:prop:reg:R4:final} into \eqref{eqn:prop:reg:main:final}, we have the desired result.
\begin{align*}
&\text{Reg}(KT)\leq \frac{Nr_{M}T(4r_MLN^{1.5}+1)^2}{4\delta^2}+ 2\sqrt{K}NT(\delta+\frac{3}{2\delta})  \nonumber\\
&+ \frac{NT(2\delta+1)^2(4r_MLN^{1.5}+1)^3}{16\delta^4} + \frac{LNMr_M \pi^2(T+1)}{6} \nonumber\\
& + LNMr_M\left(1+\frac{\pi^2}{4}\right)(T+1)  \nonumber\\
&+ 2r_M (T+1)\sqrt{6LNMS_{\max}KT\log (KT)}
\end{align*}

\section{Proof of Lemma \ref{lemma:drift}}
\label{App:pf:lemma:drift}
In the rest of the proof, we omit the frame index $kT$ without introducing any confusion. 
\begin{align}
\label{eqn:prop:vol:V}
&\Delta V(kT)\triangleq\bE\left[V((k+1)T)-V(kT)\middle|\mb{W}(kT)\right]\nonumber\\
=&\bE\left[\sqrt{\|\mb{Q}((k+1)T)\|^2}-\sqrt{\|\mb{Q}(kT)\|^2}\middle|\mb{W}(kT)\right]\nonumber\\
\leq&\frac{1}{2\|\mb{Q}(kT)\|}\underbrace{\bE\left[\|\mb{Q}((k+1)T)\|^2-\|\mb{Q}(kT)\|^2\middle|\mb{W}(kT)\right]}_{\triangleq\Delta V_1(kT)},
\end{align}
where the last step follows from the fact that $f(x)=\sqrt{x}$ is concave for $x>0$ and thus $f(x_1)-f(x_2)\leq f'(x_2)(x_1-x_2)=(x_1-x_2)/(2\sqrt{x_2})$ with $x_1=\|\mb{Q}((k+1)T)\|^2$ and $x_2=\|\mb{Q}(kT)\|^2$. Next, we consider the term $\Delta V_1(kT)$.
\begin{align}
\label{eqn:prop:vol:V1}
&\Delta V_1(kT)=\sum_{n}\bE\left[Q_n^2((k+1)T)-Q_n^2(kT)\middle|\mb{W}(kT)\right]\nonumber\\
\stackrel{(a)}{\leq}&\sum_{n}\bE\left[\left(Q_n+\lambda_n-\sum_{l}\wh{S}_{l,n}+\epsilon_k\right)^2-Q_n^2\middle|\mb{W}\right]\nonumber\\
\stackrel{(b)}{\leq}&2\sum_{n}(\lambda_n+\epsilon_k)Q_n+3N+\frac{N\delta^2}{2}-2\sum_{l,n}\bE\left[Q_n\wh{S}_{l,n}\middle|\mb{W}\right],
\end{align}
where step $(a)$ follows from the fact that $(\max\{x,0\})^2\leq x^2$ for any real number $x$; $(b)$ uses the fact that $\wh{S}_{l,n}\leq1$, $\lambda_n\leq 1$ and $\epsilon_k \leq \frac{\delta}{2}$. 

According to the definition of our proposed \hc{} algorithm, we have 
\begin{align}
&\sum_{l,n}\left(Q_n+\eta_k T \max_{m}r_mw_{l,n,m}\right)\wh{S}_{l,n}\nonumber\\
\geq&\sum_{l,n}\left(Q_n+\eta_k T\max_{m}r_mw_{l,n,m}\right)S^{\dagger}_{l,n}\nonumber\\
\geq&\sum_{l,n}Q_nS^{\dagger}_{l,n},
\end{align}
where $\mb{S}^{\dagger}\triangleq (S^{\dagger}_n)_{l,n}\in\argmax_{\mb{S}}\sum_{l,n}Q_nS_{l,n}$. Hence, we have 
\begin{align}
\label{eqn:prop:vol:V1:TLFG}
 \sum_{l,n}Q_n\wh{S}_{l,n}\geq \sum_{l,n}Q_nS^{\dagger}_{l,n}-\eta_k NLTr_M,
\end{align}
where we recall that $r_M$ is the largest available rate and use the fact that $w_{l,n,m}\leq 1, \forall n,\forall m$. Substituting \eqref{eqn:prop:vol:V1:TLFG} into \eqref{eqn:prop:vol:V1}, we have 
\begin{align}
\label{eqn:prop:vol:V1:medium}
&\Delta V_1(kT)\leq 2\sum_{n}(\lambda_n+\epsilon_k)Q_n + 3N+\frac{N\delta^2}{2}+2\eta_kNLTr_M\nonumber\\
&\qquad\qquad-2\sum_{l,n}\bE\left[Q_n(kT)S^{\dagger}_{l,n}(kT)\middle|\mb{W}(kT)\right].
\end{align}
Note that there exists non-negative numbers $\beta(\mb{s})$ with $\sum_{\mb{s}\in\mc{S}}\beta(\mb{s})=1$ satisfying
\begin{align*}
\lambda_n+\delta\leq\sum_{\mb{s}\in\mc{S}}\beta(\mb{s})\sum_{l}s_{l,n}, \forall n.
\end{align*}
Hence, we have 
\begin{align}
\label{eqn:prop:vol:V1:ranpolicy}
\sum_{n}(\lambda_n+\delta)Q_n(kT)\leq&\sum_{\mb{s}\in\mc{S}}\beta(\mb{s})\sum_{l,n}Q_ns_{l,n}\nonumber\\
\leq&\sum_{\mb{s}\in\mc{S}}\beta(\mb{s})\max_{\mb{s}\in{S}}\sum_{l,n}Q_ns_{l,n}\nonumber\\
=&\sum_{l,n}Q_nS_{l,n}^{\dagger}.
\end{align}
Substituting \eqref{eqn:prop:vol:V1:ranpolicy} into \eqref{eqn:prop:vol:V1:medium}, we have 
\begin{align}
\label{eqn:prop:vol:V1:final}
&\Delta V_1(kT)\leq -2(\delta-\epsilon_k)\sum_{n=1}^{N}Q_n+3N\nonumber\\
&\qquad\qquad\qquad+\frac{N\delta^2}{2}+2\eta_kNLTr_M\nonumber\\
\leq&-2(\delta-\epsilon_k)\|\mb{Q}\|+3N+\frac{N\delta^2}{2}+2\eta_kNLTr_M,
\end{align}
where we use the fact that $\sum_{n=1}^{N}Q_n=\|\mb{Q}\|_1\geq\|\mb{Q}\|$. Substituting \eqref{eqn:prop:vol:V1:final} into \eqref{eqn:prop:vol:V}, we have 
\begin{align*}
&\Delta V(kT)\leq\frac{1}{2\|\mb{Q}\|}\Big(-2(\delta-\epsilon_k)\|\mb{Q}\|+3N\nonumber\\
&\qquad\qquad\qquad\quad+2N\epsilon_k^2+2\eta_kNLTr_M\Big)\nonumber\\
&=-(\delta-\epsilon_k)+\frac{6N+N\delta^2+4\eta_kNLTr_M}{4V(kT)}. 
\end{align*}
This implies that for any $\epsilon_k\leq\delta/2$, if $V(kT)\geq W_{k}\triangleq 6N+N\delta^2+4\eta_kNLTr_M)/\delta$, then $\Delta V(kT)\leq-\delta/4$.

In addition, 
\begin{align}
&\left|V((k+1)T)-V(kT)\right|\nonumber\\
=&\left|\|\mb{Q}((k+1)T)\|-\|\mb{Q}(kT)\|\right|\nonumber\\
\stackrel{(a)}{\leq}&\|\mb{Q}((k+1)T)-\mb{Q}(kT)\| \nonumber\\
\stackrel{(b)}{\leq}&\|\mb{Q}((k+1)T)-\mb{Q}(kT)\|_1 \nonumber\\
\leq& N\max_n\|Q_n((k+1)T)-Q_n(kT)\|\stackrel{(c)}{\leq} 3N,
\end{align}
where step $(a)$ uses the fact that $|\|\mb{x}\|-\|\mb{y}\||\leq \|\mb{x}-\mb{y}\|$ for vectors $\mb{x}$ and $\mb{y}$; $(b)$ is true since $\|\mb{x}\|\leq\|\mb{x}\|_1$; $(c)$ is true since $\lambda_n\leq1$, $\sum_{l}\wh{S}_{l,n}(kT)\leq1$, and $\epsilon_k\leq \delta/2\leq1$.

\bibliographystyle{IEEEtran}
\bibliography{refs,refs_round4}

\end{document}